\documentclass[onecollarge]{svjour2}       
\smartqed  
\usepackage{graphicx}
\usepackage{amssymb}
\usepackage{amsmath}

\usepackage{amsthm}
\usepackage{appendix}
\usepackage[numbers]{natbib}
\usepackage{color}

\newcommand{\RR}{\mathbb{R}}      
\newcommand{\CC}{\mathbb{C}}      
\newcommand{\ZZ}{\mathbb{Z}}      
\newcommand{\vecc}{\boldsymbol}

\newtheorem{ass}{Assumption}
\begin{document}

\title{Homogenization for a Class of Generalized Langevin Equations with an Application to Thermophoresis
}



\author{Soon Hoe Lim \and Jan Wehr}


\institute{Soon Hoe Lim \at
		Program in Applied Mathematics, University of Arizona \\
              \email{shoelim@math.arizona.edu}            \\			 
          \emph{Present address:} Nordita, 
KTH Royal Institute of Technology and Stockholm University,
Roslagstullsbacken 23, SE-106 91 Stockholm, Sweden
           \and 
           Jan Wehr \at
           \email{wehr@math.arizona.edu} \\
           Department of Mathematics and Program in Applied Mathematics, University of Arizona \\
       \emph{Present address:} Department of Mathematics, University of Arizona, Tucson, AZ 85721, USA  
}

\date{\today}

\maketitle

\begin{abstract}
We study a class of systems whose dynamics are described by generalized Langevin equations with state-dependent coefficients. We find that in the limit, in which all the characteristic time scales  vanish at the same rate, the position variable of the system  converges to a homogenized process, described by an equation containing additional drift terms induced by the noise. The convergence results are obtained using the main result in \cite{hottovy2015smoluchowski}, whose version is proven here under a weaker spectral assumption on the damping matrix. We apply our  results to study thermophoresis of a Brownian particle in a non-equilibrium heat bath. 
\keywords{Generalized Langevin equation  \and small mass limit \and multiscale analysis \and noise-induced drift \and thermophoresis }
\end{abstract}

\maketitle

\section{Introduction}


From physical sciences to social sciences, many phenomena are modeled by noisy dynamical systems. In many such systems, several widely separated time scales are present. The system obtained in the homogenization limit, in which the fast time scales go to zero, is simpler than the original one, while often retaining the essential features of its dynamics \cite{majda2001mathematical, givon2004extracting,Pavliotis-TwoFast,Pavliotis}. On the other hand, the different fast time scales may compete and this competition is reflected in the homogenized equations.

Of particular interest is the model of a Brownian particle interacting with the environment  \cite{nelson1967dynamical}. The usual model for such system neglects the memory effects, representing the interaction of the particle with the environment as a sum of an instantaneous damping force and a white noise. Although such an idealized model generally gives a good approximate description of the dynamics of the particle, there are situations where the memory effects play an important role, for instance  when the particle is subject to a hydrodynamic interaction  \cite{franosch2011resonances}, or when the particle is an atom embedded in a condensed-matter heat bath \cite{Groblacher2015}.   In addition, the stochastic forcing introduced by the environment is often more accurately modeled by a colored noise than by white noise.  


In this paper, we study a class of generalized Langevin equations (GLEs), with state-dependent drift and diffusion coefficients, driven by colored noise.  They provide a realistic description of the dynamics of a classical Brownian particle in an inhomogeneous environment; their solutions are not Markov processes.  We are interested in the limiting behavior of the particle when the characteristic time scales become small and in how  competition of the time scales, as well as inhomogeneity of the environment, impact its limiting dynamics. The main mathematical result of this paper is Theorem \ref{general_result}, in which we derive the homogenization limit for a general class of non-Markovian systems. Special cases are studied in some details to obtain more explicit results. Their physical relevance is illustrated by an application to thermophoresis models.  

The paper is organized as follows. In Section \ref{nmle}, we introduce and discuss a class of GLEs, as well as its two sub-classes, to be studied in this paper.  In Section \ref{skrevisited}, we revisit the Smoluchowski-Kramers limit for a class of SDEs with state-dependent drift and diffusion coefficients, under a weaker assumption on the spectrum of the damping matrix than that used in earlier work \cite{hottovy2015smoluchowski}.  Using this result of Section \ref{skrevisited}, we study homogenization for the GLEs in Section \ref{general_homog}. We specialize the study to the two sub-classes of models in Section \ref{homog}.  In Section \ref{sec:thermophoresis}, we apply the results obtained in the previous sections to study the thermophoresis of a Brownian particle in a non-equilibrium heat bath. We end the paper by giving the conclusions and final remarks in Section \ref{conc}. The appendices provide some technical results used in the main paper, as well as physical motivation for the form of the GLEs studied here. In Appendix \ref{appA} we provide a variant of a (heuristic) derivation of the equations studied in this paper from Hamiltonian model of a particle interacting with a system of harmonic oscillators.  Appendix \ref{proof_sketch} contains a sketch of the proof of Theorem \ref{skthm}.

\section{Generalized Langevin Equations (GLEs)} \label{nmle}
\subsection{GLEs as Non-Markovian Models} \label{nmleA}

We consider a class of non-Markovian Langevin equations, with state-dependent coefficients, that describe the dynamics of a particle moving in a force field and interacting with the environment. Let $\vecc{x}_{t} \in \RR^{d}$, $t \geq 0$, be the position of the particle. The evolution of position,  $\vecc{x}_{t}$, is given by the solution to the following stochastic integro-differential equation (SIDE):
\begin{equation} \label{genle}
 m \ddot{\vecc{x}}_{t} =  \vecc{F}(\vecc{x}_{t}) - \vecc{g}(\vecc{x}_{t}) \int_{0}^{t} \vecc{\kappa}(t-s) \vecc{h}(\vecc{x}_{s}) \dot{\vecc{x}}_{s} ds +  \vecc{\sigma}(\vecc{x}_{t}) \vecc{\xi}_{t}, 
\end{equation}
with the initial conditions (here the initial time is chosen to be  $t=0$): 
\begin{equation}
\vecc{x}_{0} = \vecc{x}, \ \ \dot{\vecc{x}}_{0} = \vecc{v}.
\end{equation}
The initial conditions $\vecc{x}$ and $\vecc{v}$ are random variables independent of  the process $\{\vecc{\xi}_{t}: \ t \geq 0\}$.  Our motivation to study the SIDE \eqref{genle} is that study of microscopic dynamics leads naturally to  equations of this form (see Appendix \ref{appA}).  

Here and throughout the paper, overdot denotes derivative with respect to time $t$, the superscript $^*$ denotes conjugate transposition of matrices or vectors and $E$ denotes expectation. In the SIDE \eqref{genle}, $m > 0$ is the mass of the particle, the matrix-valued functions $\vecc{g}: \RR^{d} \to \RR^{d \times q}$, $\vecc{h} : \RR^{d} \to \RR^{q \times d}$ and $\vecc{\sigma}: \RR^{r} \to \RR^{d \times r}$ are the state-dependent coefficients of the equation, and $\vecc{F} :\RR^{d} \to \RR^{d}$ is a force field acting on the particle. Here $d$, $q$ and $r$ are, possibly distinct, positive integers. The second term on the right hand side of \eqref{genle} represents the drag experienced by the particle and the last term models the noise.  
 
The matrix-valued function $\vecc{\kappa}: \RR \to \RR^{q \times q}$ is a memory function which is {\it Bohl}, i.e. the matrix elements of $\vecc{\kappa}(t)$ are finite linear combinations of the functions of the form $t^k e^{\alpha t} \cos(\omega t)$ and $t^k e^{\alpha t} \sin(\omega t)$, where $k$ is an integer and $\alpha$ and $\omega$ are real numbers. For properties of Bohl functions, we refer to Chapter 2 of \cite{trentelman2002control}.  The noise process $\vecc{\xi}_{t}$ is a $r$-dimensional mean zero stationary real-valued Gaussian vector process having a Bohl covariance function, $\vecc{R}(t):=E \vecc{\xi}_t \vecc{\xi}_0^* = \vecc{R}^*(-t) $, and, therefore, its spectral density, $\vecc{S}(\omega) := \int_{-\infty}^{\infty} \vecc{R}(t) e^{-i\omega t} dt$, is a rational function \cite{willems1980stochastic}. 


The SIDE \eqref{genle} is a non-Markovian Langevin equation, since its solution at time $t$ depends on the entire past. Two of its terms are different than those in the usual Langevin equations. One of them is the drag term, which here involves an integral over the particle's past velocities with a memory kernel $\vecc{\kappa}(t-s)$. It describes the state-dependent dissipation which comprises the back-action effects of the environment up to current time. The other term, involving a Gaussian colored noise $\vecc{\xi}_{t}$, is a multiplicative noise term, also arising from interaction of the particle with the environment. Therefore, \eqref{genle} is a generalized Langevin equation (GLE), which in its most basic form was first introduced by Mori in \cite{mori1965transport} and subsequently used to model many systems in statistical physics \cite{Kubo_fd,toda2012statistical,goychuk2012viscoelastic}.


As remarked by van Kampen in \cite{van1998remarks}, ``Non-Markov is the rule, Markov is the exception". Therefore, it is not surprising that non-Markovian equations (including those of form \eqref{genle}) find numerous applications and thus have been studied widely in the mathematical, physical and engineering literature.  See, for instance,  \cite{luczka2005non, samorodnitsky1994stable} for surveys of non-Markovian processes, \cite{PhysRevB.89.134303,mckinley2009transient,adelman1976generalized} for physical applications and \cite{Ottobre} for asymptotic analysis.




Note that the Gaussian process $\vecc{\xi}_t$ which drives the SIDE \eqref{genle} is not assumed to be Markov.  The assumptions we made on its covariance will allow us to present it as a projection of a Markov process in a (typically higher-dimensional) space.  This approach, which originated in stochastic control theory \cite{kalman1960new}, is called {\it stochastic realization}. We describe it in detail below.

Let $\vecc{\Gamma}_1 \in \RR^{d_1 \times d_1}$, $\vecc{M}_1 \in \RR^{d_1 \times d_1}$, $\vecc{C}_1 \in \RR^{q \times d_1}$, $\vecc{\Sigma}_1 \in \RR^{d_1 \times q_1}$, $\vecc{\Gamma}_2 \in \RR^{d_2 \times d_2}$, $\vecc{M}_2 \in \RR^{d_2 \times d_2}$, $\vecc{C}_2 \in \RR^{r \times d_2}$, $\vecc{\Sigma}_2 \in \RR^{d_2 \times q_2}$ be constant matrices, where $d_1,d_2,q_1,q_2$, $q$ and $r$ are positive integers. 
In this paper, we study the class of SIDE \eqref{genle}, with the memory function defined in terms of the triple $(\vecc{\Gamma}_1,\vecc{M}_1,\vecc{C}_1)$ of matrices as follows:
\begin{equation} \label{memory_realized}
\vecc{\kappa}(t)=\vecc{C}_1e^{-\vecc{\Gamma_1}|t|}\vecc{M}_1\vecc{C}_1^*.
\end{equation}
The noise process is the mean zero, stationary Gaussian vector process,  whose covariance will be expressed in terms of the triple $(\vecc{\Gamma}_2,\vecc{M}_2,\vecc{C}_2)$.  More precisely, we define it as:
\begin{equation} \label{noise}
\vecc{\xi}_t = \vecc{C}_2 \vecc{\beta}_t,\end{equation}
where $\vecc{\beta}_t$ is the solution to the It\^o SDE:
\begin{equation} \label{realize}
d\vecc{\beta}_t = -\vecc{\Gamma}_2\vecc{\beta}_t dt + \vecc{\Sigma}_2 d\vecc{W}^{(q_2)}_t,
\end{equation}
with the initial condition, $\vecc{\beta}_0$, normally distributed with zero mean and covariance  $\vecc{M}_2$. Here, $\vecc{W}^{(q_2)}_t$ denotes a $q_2$-dimensional Wiener process and is independent of $\vecc{\beta}$. Throughout the paper the dimension of the Wiener process will be specified by the superscript.

For $i=1,2$, the matrix $\vecc{\Gamma}_i$ is {\it positive stable}, i.e. all its eigenvalues have positive real parts and $\vecc{M}_i = \vecc{M}_i^* > 0$ satisfies the following Lyapunov equation:
\begin{equation} 
\vecc{\Gamma}_i \vecc{M}_i+\vecc{M}_i \vecc{\Gamma}_i^*=\vecc{\Sigma}_i \vecc{\Sigma}_i^*.
\end{equation}
It follows from positive stability of $\vecc{\Gamma}_i$ that this equation indeed has a unique solution \cite{bellman1997introduction}. 
The covariance matrix, $\vecc{R}(t) \in \RR^{r \times r}$, of the noise process is therefore expressed in terms of  the matrices $(\vecc{\Gamma}_2,\vecc{M}_2,\vecc{C}_2)$ as follows:
\begin{equation} \label{cov}
\vecc{R}(t)=\vecc{C}_2e^{-\vecc{\Gamma_2}|t|}\vecc{M}_2\vecc{C}_2^*, 
\end{equation}
and therefore the triple $(\vecc{\Gamma}_2,\vecc{M}_2,\vecc{C}_2)$ completely specifies the probability distribution of $\vecc{\xi}_t$. It is worth mentioning  that the triples that specify the memory function in \eqref{memory_realized} and the noise process in \eqref{noise} are only unique up to the following transformations:  
\begin{equation} \label{transf_realize}
(\vecc{\Gamma}'_i=\vecc{T}_i \vecc{\Gamma}_i \vecc{T}^{-1}_i, \vecc{M}_i' = \vecc{T}_i \vecc{M}_i \vecc{T}_i^{*}, \vecc{C}'_i =  \vecc{C}_i \vecc{T}_i^{-1}),
\end{equation}
where $i=1,2$ and $\vecc{T}_i$ is any invertible matrices of appropriate dimensions.

The triple $(\vecc{\Gamma}_2,\vecc{M}_2,\vecc{C}_2)$ above is called a {\it (weak) stochastic realization} of the covariance matrix $\vecc{R}(t)$ in the well established theory of stochastic realization, which is concerned with solving the inverse problem of stationary covariance generation (see \cite{lindquist1985realization,lindquist2015linear}).  Any zero mean stationary Gaussian process, $\vecc{\xi}'_t$, having a Bohl covariance function, can be realized as a projection of a Gaussian Markov process in the above way.  Let us remark that  Gaussian processes with Bohl covariance functions are precisely those with rational spectral density \cite{willems1980stochastic}.  


Our approach allows to consider the most general Gaussian noises that can be realized in a finite-dimensional state space in the above way (i.e. as a linear transformation of a Gaussian Markov process).  In fact, the condition on the covariance function to have entries in the Bohl class is necessary and sufficient for solvability of the problem of stochastic realization of stationary Gaussian processes.  We refer to the propositions and theorems on page 303-308 of \cite{willems1980stochastic} for a brief exposition of  stochastic realization problems.


\begin{remark} Physically, the choice of the matrices  $\vecc{\Gamma}_2,\vecc{M}_2,\vecc{C}_2$ specifies the characteristic time scales (eigenvalues of $\vecc{\Gamma}_2^{-1}$) present in the environment, introduces the initial state of a stationary Markovian Gaussian noise and selects the parts of the prepared Markovian noise that are (partially) observed, respectively.  In other words, we have assumed that the  noise in the SIDE \eqref{genle} is realized or ``experimentally prepared" by the above triple of matrices. 
\end{remark}

For our homogenization study of the equation \eqref{genle} we need the {\it effective damping constant}, 
\begin{equation} \label{eff_damping}
\vecc{K}_1 := \int_0^{\infty} \vecc{\kappa}(t) dt = \vecc{C}_1 \vecc{\Gamma}_1^{-1} \vecc{M}_1 \vecc{C}_1^* \in \RR^{q \times q},
\end{equation}
and the {\it effective diffusion constant}, 
\begin{equation} \label{eff_diff}
\vecc{K}_2 := \int_0^{\infty} \vecc{R}(t) dt = \vecc{C}_2 \vecc{\Gamma}_2^{-1} \vecc{M}_2 \vecc{C}_2^* \in \RR^{r \times r},
\end{equation}
to be invertible (see Section \ref{nmleB}). This is equivalent to the matrices $\vecc{C}_i$ having full rank. Homogenization for a class of systems with vanishing effective damping and/or diffusion constant \cite{bao2005non} will be explored in our future work.

With the above definitions of memory kernel and noise process, the SIDE \eqref{genle} becomes:
\begin{equation} \label{genle_general}
 m \ddot{\vecc{x}}_{t} =  \vecc{F}(\vecc{x}_{t}) - \vecc{g}(\vecc{x}_{t}) \int_{0}^{t} \vecc{C}_1 e^{-\vecc{\Gamma}_1(t-s)} \vecc{M}_1\vecc{C}_1^* \vecc{h}(\vecc{x}_{s}) \dot{\vecc{x}}_{s} ds +  \vecc{\sigma}(\vecc{x}_{t}) \vecc{C}_2 \vecc{\beta}_{t}, 
\end{equation}
where $\vecc{\beta}_t$ is the solution to the SDE \eqref{realize}.
To illustrate the results of the general study in important special cases (which will also be used later in applications), we consider two representative classes of SIDE \eqref{genle}. The driving Gaussian colored noise is Markovian in the first class and non-Markovian in the second. We set $d=d_1=d_2=q_1=q_2=q=r$ in the following  examples. 

\begin{itemize}
\item[(i)] {\it Example of a SIDE driven by a Markovian colored noise.}  The memory kernel is given by an exponential function, i.e. 
\begin{equation}
\vecc{\kappa}(t-s) = \vecc{\kappa}_{1}(t-s) := \vecc{A} e^{-\vecc{A}|t-s|},\end{equation} 
where $\vecc{A} \in \RR^{d \times d}$ is a constant diagonal matrix with positive eigenvalues.  The driving noise is  the Ornstein-Uhlenbeck (OU) process, $\vecc{\xi}_{t} = \vecc{\eta}_{t} \in \RR^d$, i.e. a mean zero stationary Gaussian process which is the solution to the SDE: 
\begin{equation} \label{ou}
d\vecc{\eta}_{t} = -\vecc{A} \vecc{\eta}_{t} dt + \vecc{A} d\vecc{W}^{(d)}_{t}.
\end{equation}
In order for the process $\vecc{\eta}_{t}$ to be stationary, the initial condition has to be distributed according to the (unique) stationary measure of the Markov process defined by the above equation, i.e. $\vecc{\eta}_{0} = \vecc{\eta}$ is normally distributed with zero mean and covariance $\vecc{A}/2$. The mean and the covariance of $\vecc{\eta}_{t}$ are given by:
\begin{equation} \label{o-u_stats}
E[\vecc{\eta}_{t}] = 0, \ \ E[\vecc{\eta}_{t} \vecc{\eta}_{s}^{*}] = \frac{1}{2}\vecc{\kappa}_{1}(t-s), \ \ s,t \geq 0.
\end{equation}

 The resulting SIDE reads:
\begin{equation} \label{side2}
m \ddot{\vecc{x}}_{t} =  \vecc{F}(\vecc{x}_{t})  - \vecc{g}(\vecc{x}_{t})\int_{0}^{t} \vecc{\kappa}_{1}(t-s) \vecc{h}(\vecc{x}_{s}) \dot{\vecc{x}}_{s} ds + \vecc{\sigma}(\vecc{x}_{t}) \vecc{\eta}_{t}.
\end{equation}
Let us note that Ornstein-Uhlenbeck processes are the only stationary, ergodic, Gaussian, Markov processes with continuous covariance functions \cite{pavliotis2014stochastic}.  When all diagonal entries of $\vecc{A}$ go to infinity, the OU process approaches the white noise.  For details on OU processes, see for instance  \cite{pavliotis2014stochastic} and Section 2 of \cite{hottovy2015small}.
\item[(ii)] {\it Example of a SIDE driven by a non-Markovian colored noise.} The memory kernel is given by an oscillatory function whose amplitude is exponentially decaying, i.e. $ \vecc{\kappa}(t-s) := \vecc{\kappa}_2(t-s)$, a diagonal matrix with the diagonal entries:
\begin{equation}\label{harmonic_memory}
(\vecc{\kappa}_{2})_{ii}(t-s) :=
\begin{cases}
    \frac{1}{\tau_{ii}} e^{-\omega_{ii}^2  \frac{|t-s|}{2 \tau_{ii}}}\left[\cos\left(\frac{\omega^0_{ii}}{\tau_{ii}} (t-s) \right) + \frac{\omega^1_{ii}}{2} \sin\left(\frac{\omega^0_{ii}}{\tau_{ii}} |t-s| \right) \right], & \text{if } |\omega_{ii}|<2 \\
  \frac{1}{\tau_{ii}} e^{-\omega_{ii}^2  \frac{|t-s|}{2 \tau_{ii}}}\left[\cosh\left(\frac{\tilde{\omega}^0_{ii}}{\tau_{ii}}(t-s) \right) + \frac{\tilde{\omega}^1_{ii}}{2} \sinh\left(\frac{\tilde{\omega}^0_{ii}}{\tau_{ii}} |t-s| \right) \right],  & \text{if } |\omega_{ii}|>2, 
\end{cases}
\end{equation}
where, for $i=1,\dots,d$, $\tau_{ii}$ is a positive constant, $\omega_{ii}$ is a real constant, $\omega^0_{ii} := \omega_{ii}\sqrt{1-\omega_{ii}^2/4}$, $\tilde{\omega}^0_{ii} := \omega_{ii}\sqrt{\omega_{ii}^2/4-1}$, $\omega_{ii}^1 := \omega_{ii}/\sqrt{1-\omega_{ii}^2/4}$, and $\tilde{\omega}_{ii}^1 := \omega_{ii}/\sqrt{\omega_{ii}^2/4-1}$.  

Let  $\vecc{\tau}$ be  constant diagonal matrix with the positive eigenvalues $(\tau_{jj})_{j=1}^d$, $\vecc{\Omega}$ be  constant diagonal matrix with the real eigenvalues $(\omega_{jj})_{j=1}^{d}$, $\vecc{\Omega}_{0}$  be constant $d \times d$ diagonal matrix with the eigenvalues $\omega_{jj}\sqrt{1-\omega_{jj}^2/4}$ (if $|\omega_{jj}| < 2$) and  $i \omega_{jj}\sqrt{\omega_{jj}^2/4-1}$ (if $|\omega_{jj}|>2$), and $\vecc{\Omega}_{1}$ be  constant $d \times d$ diagonal matrix with the eigenvalues $\omega_{jj}/\sqrt{1-\omega_{jj}^2/4}$ (if $|\omega_{jj}|<2$) and $-i\omega_{jj}/\sqrt{\omega_{jj}^2/4-1}$ (if $|\omega_{jj}|>2$), where $i$ is the imaginary unit.

The driving noise is given by the harmonic noise process, $\vecc{\xi}_{t}=\vecc{h}_{t} \in \RR^d$, i.e. a mean zero stationary Gaussian process which is the solution to the SDE system: 
\begin{align}
\vecc{\tau} d\vecc{h}_{t} &=  \vecc{u}_{t} dt, \label{unscaled_har1} \\ 
\vecc{\tau} d\vecc{u}_{t} &= -\vecc{\Omega}^2 \vecc{u}_{t} dt - \vecc{\Omega}^2 \vecc{h}_{t} dt + \vecc{\Omega}^2  d\vecc{W}^{(d)}_{t}, \label{unscaled_har2}
\end{align}
with the initial conditions, $\vecc{h}_0$ and $\vecc{u}_0$, distributed according to the (unique) stationary measure of the above SDE system. The mean and the covariance of $\vecc{h}_{t}$ are given by: 
\begin{equation}
E[\vecc{h}_{t}] =  \vecc{0}, \ \ E[\vecc{h}_{t} \vecc{h}_{s}^{*}] = \frac{1}{2} \vecc{\kappa}_{2}(t-s),\ \  s, t \geq 0.\end{equation}
Note that $\vecc{h}_t$ is not a Markov process (but the process $(\vecc{h}_t, \vecc{u}_t)$ is).

The resulting SIDE reads:
\begin{equation} \label{side3}
m \ddot{\vecc{x}}_{t} =  \vecc{F}(\vecc{x}_{t}) - \vecc{g}(\vecc{x}_{t})\int_{0}^{t} \vecc{\kappa}_{2}(t-s) \vecc{h}(\vecc{x}_{s}) \dot{\vecc{x}}_{s} ds + \vecc{\sigma}(\vecc{x}_{t}) \vecc{h}_{t}.
\end{equation}

The harmonic noise is  an approximation of the white noise, smoother than  the Ornstein-Uhlenbeck process. It can be shown that in the limit $\omega_{ii} \to \infty$ (for all $i$) the process $\vecc{h}_{t}$ converges to the Ornstein-Uhlenbeck process whose $i$th component process has correlation time $\tau_{ii}$, whereas in the limit $\tau_{ii} \to 0$ (for all $i$) the process $\vecc{h}_{t}$ converges to the white noise. For detailed properties of harmonic noise process, see for instance \cite{schimansky1990harmonic} or the Appendix in \cite{McDaniel14}. We remark that the harmonic noise is one of the simplest examples of a non-Markovian process and its use as the driving  noise in the SIDE \eqref{genle} is a natural choice that models the environment as a bath of damped harmonic oscillators \cite{hanggi1993can}. 
\end{itemize}

\begin{remark} \label{dimofnoise} Note that in the SIDEs for the above two sub-classes, the dimension of the driving Wiener process, $\vecc{W}_t^{(d)}$, is the same as that of the colored noise processes $\vecc{\eta}_t$ and $\vecc{h}_t$, as well as the processes, $\vecc{x}_t$ and $\vecc{v}_t$. One could as well consider realizing the noise processes using a driving Wiener process of different dimension. Our choice of working with the same dimensions is for the sake of convenience as it will help to simplify the exposition. 
\end{remark}

\begin{remark}
Without loss of generality (due to \eqref{transf_realize}), we have taken $\vecc{A}$ and $\vecc{\Omega}$ to be diagonal.  
\end{remark}

\begin{remark}
In cases of particular interest in statistical physics, the triples $(\vecc{\Gamma}_i, \vecc{M}_i, \vecc{C}_i)$ coincide, up to the transformations in \eqref{transf_realize}, for $i=1,2$; $\vecc{h} = \vecc{g}^*$ and $\vecc{g}$ is proportional to $\vecc{\sigma}$, with the proportionality factor equals $k_B T$, where $k_B$ denotes the Boltzmann constant and $T > 0$ is temperature of the environment (see Appendix \ref{appA}). In this case, we have $d_1=d_2$ and $q=r$.  In particular, for the two sub-classes above we have
\begin{equation} \label{fdt_sc1}
E[\vecc{\eta}_{t}^{0} (\vecc{\eta}_{s}^{0})^{*}] = \frac{1}{2} \vecc{\kappa}_1(t-s)
\end{equation}
for the first sub-class and 
\begin{equation}
E[\vecc{h}_{t}^{0} (\vecc{h}_{s}^{0})^{*}] = \frac{1}{2} \vecc{\kappa}_{2}(t-s)
\end{equation}
for the second sub-class. 
In such cases, the SIDEs describe a particle interacting with an equilibrium heat bath at a temperature $T$, whose dynamics satisfy the fluctuation-dissipation relation \cite{toda2012statistical,zwanzig2001nonequilibrium}.
\end{remark}

\subsection{Homogenization of SIDEs: Discussion and  Statement of the Problem} \label{nmleB}

There are three characteristic time scales defining the non-Markovian dynamics described by the SIDE \eqref{genle}:
\begin{itemize}
\item[(i)] the inertial time scale, $\lambda_{m}$, proportional to $m$, whose physical significance is the relaxation time of the particle velocity process $\vecc{v}_{t} := \dot{\vecc{x}}_{t}$. The limit $\lambda_{m} \to 0$ is equivalent to the limit $m \to 0$; 
\item[(ii)] the memory time scale, $\lambda_{\kappa}$, defined as the inverse rate of exponential decay of the memory kernel $\vecc{\kappa}(t-s)$;
\item[(iii)] the noise correlation time scale, $\lambda_{\xi}$.
\end{itemize}

For the purpose of general multiscale analysis, we set $m = m_{0} \epsilon^{\mu}$, $\lambda_{\kappa} = \tau_{\kappa}  \epsilon^{\theta}$ and $\lambda_{\xi} = \tau_{\xi} \epsilon^{\nu}$, 
 where $\epsilon > 0$ is a parameter which will be taken to zero, $m_0$, $\tau_\kappa$, $\tau_\xi$ are (fixed) proportionality constants, and $\mu, \theta, \nu$ are positive constants (exponents), specifying the orders at which the time scales $\lambda_{m}, \lambda_{\kappa}, \lambda_{\xi}$ vanish  respectively. We consider a family of SIDEs, parametrized by $\epsilon$, with the inertial time scale $\lambda_m$ proportional to $m_0 \epsilon^\mu$, memory time scale $\lambda_{\kappa} = \tau_\kappa \epsilon^\theta$ and noise correlation time scale $\lambda_{\xi} = \tau_{\xi} \epsilon^{\nu}$, to be defined in the following.

We replace $m$ with $m_0 \epsilon^\mu$,  $\vecc{\Gamma}_1$ with $\vecc{\Gamma}_1/(\tau_{\kappa} \epsilon^{\theta})$, $\vecc{M}_1$ with $\vecc{M}_1/(\tau_{\kappa} \epsilon^{\theta})$,  and $\vecc{x}_t$ with $\vecc{x}_t^\epsilon$ in \eqref{genle_general}. Also, we substitute  $\vecc{\Gamma}_2$ with $\vecc{\Gamma}_2/(\tau_{\xi} \epsilon^{\nu})$, $\vecc{\Sigma}_2$ with $\vecc{\Sigma}_2/(\tau_{\xi} \epsilon^{\nu})$, and $\vecc{\beta}_t$ with $\vecc{\beta}_t^\epsilon$  in \eqref{realize}. The SIDE \eqref{genle_general} then becomes: 
\begin{equation} \label{general_side_rescaled}
m_0 \epsilon^{\mu} \ddot{\vecc{x}}^{\epsilon}_{t} =  \vecc{F}(\vecc{x}^\epsilon_{t}) - \frac{\vecc{g}(\vecc{x}^\epsilon_{t})}{\tau_{\kappa} \epsilon^{\theta}} \int_{0}^{t} \vecc{C}_1e^{-\frac{\vecc{\Gamma}_1}{\tau_{\kappa}\epsilon^{\theta}}(t-s)} \vecc{M}_1 \vecc{C}_1^* \vecc{h}(\vecc{x}^\epsilon_{s}) \dot{\vecc{x}}^{\epsilon}_{s} ds +  \vecc{\sigma}(\vecc{x}^\epsilon_{t}) \vecc{C}_2 \vecc{\beta}^\epsilon_{t}, 
\end{equation}
with the initial conditions, $\vecc{x}^\epsilon_0 = \vecc{x}$ and $\vecc{v}^\epsilon_0 = \vecc{v}$, where $\vecc{\beta}^\epsilon_t$ is a process, with correlation time $\tau_{\xi} \epsilon^\nu$,  satisfying the SDE: 
\begin{equation} \label{general_rescaled_ou}
d\vecc{\beta}^\epsilon_t = -\frac{\vecc{\Gamma}_2}{\tau_{\xi} \epsilon^\nu} \vecc{\beta}^\epsilon_t dt + \frac{\vecc{\Sigma}_2}{\tau_{\xi} \epsilon^\nu} d\vecc{W}^{(q_2)}_t,
\end{equation}
with the initial condition, $\vecc{\beta}^\epsilon_0$, normally distributed with zero mean and covariance of $\vecc{M}_2/(\tau_{\xi} \epsilon^\nu)$. 

We will also perform similar analysis on the two sub-classes of SIDE, in which case:
\begin{itemize}
\item[(i)] the SIDE \eqref{side2} becomes (with $m:=m_0 \epsilon^\mu$, $\vecc{A}$ in the formula for $\vecc{\kappa}_{1}$ replaced by $\vecc{A}/(\tau_{\kappa} \epsilon^{\theta})$,  $\vecc{A}$ in $\eqref{ou}$ replaced by $\vecc{A}/(\tau_{\eta}\epsilon^{\nu})$, $\vecc{x}_t$ replaced by $\vecc{x}_t^\epsilon$ and $\vecc{\eta}_{t}$ replaced by $\vecc{\eta}_{t}^\epsilon$):  
\begin{equation} \label{goal2}
m_{0} \epsilon^{\mu} \ddot{\vecc{x}}^{\epsilon}_{t} = \vecc{F}(\vecc{x}^\epsilon_{t}) - \frac{\vecc{g}(\vecc{x}^\epsilon_{t})}{\tau_{\kappa} \epsilon^{\theta}} \int_{0}^{t} \vecc{A} e^{-\frac{\vecc{A}}{\tau_{\kappa} \epsilon^{\theta}} (t-s)} \vecc{h}(\vecc{x}^\epsilon_{s}) \dot{\vecc{x}}^{\epsilon}_{s} ds +  \vecc{\sigma}(\vecc{x}^\epsilon_{t}) \vecc{\eta}^\epsilon_{t},
\end{equation} 
where  $\vecc{\eta}^\epsilon_{t}$ is the Ornstein-Uhlenbeck process with the correlation time  $\tau_{\eta} \epsilon^{\nu}$, i.e. it is a  process satisfying the SDE:
\begin{equation} \label{rescaled-ou}
d\vecc{\eta}^\epsilon_{t} = -\frac{\vecc{A}}{\tau_{\eta} \epsilon^{\nu}}  \vecc{\eta}^\epsilon_{t} dt + \frac{\vecc{A}}{\tau_{\eta} \epsilon^{\nu}} d \vecc{W}^{(d)}_{t}. 
\end{equation}
\item[(ii)] the SIDE \eqref{side3} becomes (with $m:=m_0 \epsilon^\mu$, $\tau_{ii} := \tau_{\kappa} \epsilon^{\theta}$ in the formula for ($\vecc{\kappa}_{2})_{ii}$ in \eqref{harmonic_memory}, $\vecc{x}_t$ replaced by $\vecc{x}^\epsilon_t$,  $\vecc{h}_{t}$, $\vecc{u}_{t}$ replaced by $\vecc{h}^\epsilon_{t}$, $\vecc{u}^\epsilon_{t}$ respectively and $\vecc{\tau} := \tau_{h} \epsilon^{\nu} \vecc{I}$ in $\eqref{unscaled_har1}$-$\eqref{unscaled_har2}$): 
\begin{align} 
&m_{0} \epsilon^{\mu} \ddot{\vecc{x}}^\epsilon_{t} \nonumber \\ 
&= \vecc{F}(\vecc{x}^\epsilon_{t})  
- \frac{\vecc{g}(\vecc{x}^\epsilon_{t})}{\tau_{\kappa} \epsilon^{\theta}} \int_{0}^{t} e^{-\vecc{\Omega}^2\frac{(t-s)}{2\tau_{\kappa}\epsilon^{\theta}}}\left[\cos\left(\frac{\vecc{\Omega}_{0}}{\tau_{\kappa}\epsilon^{\theta}}(t-s) \right) + \frac{\vecc{\Omega}_{1}}{2} \sin\left(\frac{\vecc{\Omega}_{0}}{\tau_{\kappa}\epsilon^{\theta}}(t-s) \right) \right] \vecc{h}(\vecc{x}^\epsilon_{s}) \dot{\vecc{x}}^{\epsilon}_{s} ds  +  \vecc{\sigma}(\vecc{x}^\epsilon_{t}) \vecc{h}^\epsilon_{t}, \label{goal3}
\end{align} 
where  $\vecc{h}^\epsilon_{t}$ is the harmonic noise process with the correlation time $\tau_{h} \epsilon^{\nu}$, i.e. it is a  process satisfying the SDE system:
\begin{align} 
d\vecc{h}^\epsilon_{t} &= \frac{1}{\tau_{h}\epsilon^{\nu}} \vecc{u}^\epsilon_{t} dt, \label{rescaled_h1} \\ 
d\vecc{u}^\epsilon_{t} &= -\frac{\vecc{\Omega}^2}{\tau_{h} \epsilon^{\nu}} \vecc{u}^\epsilon_{t} dt - \frac{\vecc{\Omega}^2}{\tau_{h} \epsilon^{\nu}} \vecc{h}^\epsilon_{t} dt + \frac{\vecc{\Omega}^2}{\tau_{h} \epsilon^{\nu}} d\vecc{W}^{(d)}_{t}. \label{rescaled_h2}
\end{align}

Both SIDEs have the initial conditions $\vecc{x}^\epsilon_{0} = \vecc{x}, \  \dot{\vecc{x}}^\epsilon_{0} = \vecc{v}$. The initial conditions $\vecc{\eta}^\epsilon_{0}$ (respectively, $\vecc{h}^\epsilon_{0}$ and $\vecc{u}^\epsilon_{0}$) are distributed according to the stationary measure of the SDE that the process $\vecc{\eta}^\epsilon_{t}$ (respectively, $\vecc{h}^\epsilon_{t}$ and $\vecc{u}^\epsilon_{t}$) satisfies.
\end{itemize}

In this paper we set $\mu = \theta = \nu$, which is the case when all the characteristic time scales are of comparable magnitude in the limit as $\epsilon \to 0$.  Our main goal is to derive a limiting equation for the (slow) $\vecc{x}^\epsilon$-component of the process solving the equations \eqref{general_side_rescaled}-\eqref{general_rescaled_ou}, including the special cases  $\eqref{goal2}$-$\eqref{rescaled-ou}$ and $\eqref{goal3}$-$\eqref{rescaled_h2}$, in the limit as  $\epsilon \to 0$, in a strong pathwise sense.

We explain the motivation behind the above rescalings. The rescaling of the memory kernels $\vecc{\kappa}(t-s)$, $\vecc{\kappa}_{1}(t-s)$, $\vecc{\kappa}_{2}(t-s)$ is such that in the limit $\epsilon \to 0$ the rescaled memory kernels converge to $\vecc{K}_1 \delta(t) $  formally, where $\delta(t)$ is the Dirac-delta function and $\vecc{K}_1$ is the effective damping constant defined in \eqref{eff_damping}. On the other hand, the noise processes $\vecc{\xi}^\epsilon_t = \vecc{C}_2 \vecc{\beta}^\epsilon_{t}$, $\vecc{\eta}^\epsilon_{t}$ and $\vecc{h}^\epsilon_{t}$ converge to  white noise processes in the limit $\epsilon \to 0$. 



\section{Smoluchowski-Kramers Limit of SDE's  Revisited} \label{skrevisited}
Let $(\vecc{x}_{t}^m, \vecc{v}^m_{t}) \in \RR^{n} \times \RR^n$, where $t \in [0,T]$, be a family of solutions (parametrized by a positive constant $m$) to the following SDEs:
\begin{align}
d\vecc{x}^m_{t} &= \vecc{v}^m_{t} dt,  \label{gsk1} \\ 
m d\vecc{v}^m_{t} &= \vecc{F}(\vecc{x}^m_{t}) dt -\boldsymbol{\gamma}(\vecc{x}^m_{t}) \vecc{v}^m_{t} dt  + \boldsymbol{\sigma}(\vecc{x}^m_{t}) d\boldsymbol{W}^{(k)}_{t}.\label{gsk2}
\end{align}
In the SDEs above,  $m > 0$ is the mass of the particle, $\vecc{F}: \RR^{n} \to \RR^{n}$, $\boldsymbol{\gamma}: \RR^{n} \to \RR^{n \times n}$, $\boldsymbol{\sigma}: \RR^{n} \to \RR^{n \times k}$, and $\vecc{W}^{(k)}$ is a $k$-dimensional Wiener process on the filtered probability space $(\Omega, \mathcal{F}, \mathcal{F}_{t}, \mathbb{P})$ satisfying the usual conditions \citep{karatzas2012Brownian}. 
The initial conditions are given by $\vecc{x}^m_{0} = \vecc{x}, \ \vecc{v}^m_{0} = \vecc{v}^m$. The above SDE system models diffusive phenomena in cases where the damping coefficient $\vecc{\gamma}$ and diffusion coefficient $\vecc{\sigma}$ are state-dependent. 

The Smoluchowski-Kramers limit (or the small mass limit) of the system \eqref{gsk1}-\eqref{gsk2}  has been studied in \citep{hottovy2015smoluchowski, Hottovy12, Herzog2016, birrell2017small,birrell2017homogenization}. The main result in \citep{birrell2017homogenization} says that, under certain assumptions, the $\vecc{x}^m$-component of the solution to \eqref{gsk1}-\eqref{gsk2} converges (in a strong pathwise sense), as $m \to 0$, to the solution of a homogenized SDE that contains in particular the so-called noise-induced drift, that was not present in the pre-limit SDEs (see Theorem \ref{skthm} for a precise statement). The presence of such noise-induced drift in the homogenized equation is a consequence of the state-dependence of the damping coefficient (and therefore also the diffusion coefficient if the system satisfies a fluctuation-dissipation relation). For an overview of the noise-induced drift phenomena, we refer to the review article \citep{volpe2016effective}.  

In all the works mentioned previously, the spectral assumption made on the matrix $\boldsymbol{\gamma}$ was that the symmetrized damping matrix  $\frac{1}{2}(\boldsymbol{\gamma} + \boldsymbol{\gamma}^{*})$ is uniformly positive definite (i.e. its smallest eigevalue is positive uniformly in $\vecc{x}$). The same results can be obtained under a weaker assumption that matrix $\boldsymbol{\gamma}$ is {\it uniformly positive stable}, i.e. all real parts of the eigenvalues of $\boldsymbol{\gamma}$ are  positive uniformly in $\vecc{x}$ \citep{horn1994topics}. \\


\noindent {\bf Notation.} Here and in the following, we use Einstein summation convention on repeated indices. The Euclidean norm of a vector $\vecc{w}$ is denoted by $| \vecc{w} |$ and the (induced operator) norm of a matrix $\vecc{A}$ by $\| \vecc{A} \|$. For an $\RR^{n_2 \times n_3}$-valued function $\vecc{f}(\vecc{y}):=([f]_{jk}(\vecc{y}))_{j=1,\dots,n_2; k=1,\dots, n_3}$, $\vecc{y} := ([y]_1, \dots, [y]_{n_1}) \in \RR^{n_1}$, we denote by $(\vecc{f})_{\vecc{y}}(\vecc{y})$ the $n_1 n_2 \times n_3$ matrix:
\begin{equation}
(\vecc{f})_{\vecc{y}}(\vecc{y}) = (\vecc{\nabla}_{\vecc{y}}[f]_{jk}(\vecc{y}))_{j=1,\dots, n_2; k=1,\dots,n_3},  
\end{equation}
where $\vecc{\nabla}_{\vecc{y}}[f]_{jk}(\vecc{y})$ denotes the gradient vector $(\frac{\partial [f]_{jk}(\vecc{y})}{\partial [y]_1}, \dots, \frac{\partial [f]_{jk}(\vecc{y})}{\partial [y]_{n_1}}) \in \RR^{n_1}$ for every $j,k$.  
The symbol $\mathbb{E}$ denotes expectation with respect to $\mathbb{P}$. \\

We make the following assumptions. 

\begin{ass} \label{a1} For every $\vecc{x} \in \RR^n$, the functions $\boldsymbol{F}(\vecc{x})$ and $\boldsymbol{\sigma}(\vecc{x})$ are continuous, bounded and Lipschitz in $\vecc{x}$, whereas the functions $\boldsymbol{\gamma}(\vecc{x})$ and $(\vecc{\gamma})_{\vecc{x}}(\vecc{x})$ are continuously differentiable, bounded and Lipschitz in $\vecc{x}$. Moreover, $(\vecc{\gamma})_{\vecc{x} \vecc{x}}(\vecc{x})$ is bounded for every $\vecc{x} \in \RR^n$. 
\end{ass}

\begin{ass} \label{a2} The  matrix $\boldsymbol{\gamma}$ is {\it uniformly positive stable}, i.e. all real parts of the eigenvalues of $\boldsymbol{\gamma}(\vecc{x})$ are bounded below by some constant $2\kappa > 0$, uniformly in $\vecc{x}\in \RR^n$.
\end{ass}

\begin{ass} \label{a3} The initial condition $\vecc{x}^m_0 = \vecc{x}_0$ is a random variable independent of $m$ and has finite moments of all orders,  i.e. $\mathbb{E} |\vecc{x}|^{p} < \infty$ for all $p > 0$.   The initial condition $\vecc{v}^m_0$ is a  random variable that possibly depends on $m$ and we assume that for every $p>0$, $\mathbb{E} |m \vecc{v}^m|^p = O(m^\alpha)$ as $m \to 0$, for some $\alpha \geq p/2$. 
\end{ass}

\begin{ass} \label{a4} The global solutions, defined on $[0,T]$, to the pre-limit SDEs \eqref{gsk1}-\eqref{gsk2} and to the limiting SDE \eqref{sklim} a.s. exist and are unique for all $m > 0$ (i.e. there are no explosions).  
\end{ass}


We now state the result.

\begin{theorem} \label{skthm} Suppose that the SDE system $\eqref{gsk1}$-$\eqref{gsk2}$ satisfies Assumption $\ref{a1}$-$\ref{a4}$. Let $(\vecc{x}^m_{t}, \vecc{v}^m_{t}) \in \RR^{n} \times \RR^{n}$ be its solution, with the initial condition $(\vecc{x}, \vecc{v}^m)$. Let $\vecc{X}_{t} \in \RR^{n}$ be the solution to the following It\^o SDE with initial position $\vecc{X}_{0} = \vecc{x}$: 
\begin{equation} \label{sklim}
d\vecc{X}_{t} = [\boldsymbol{\gamma}^{-1}(\vecc{X}_{t}) \vecc{F}(\vecc{X}_{t}) + \vecc{S}(\vecc{X}_{t})] dt + \boldsymbol{\gamma}^{-1}(\vecc{X}_{t}) \boldsymbol{\sigma}(\vecc{X}_{t}) d \vecc{W}^{(k)}_{t}, 
\end{equation}
where $\vecc{S}(\vecc{X}_{t})$ is the noise-induced drift whose $i$th component is given by 
\begin{equation}
S_{i}(\vecc{X}) = \frac{\partial}{\partial X_{l}}[(\gamma^{-1})_{ij}(\vecc{X})] J_{jl}(\vecc{X}), \ \ i,j,l = 1, \dots, n, \end{equation}
and $\vecc{J}$ is the unique matrix solving the Lyapunov equation 
\begin{equation} \label{lyp}
\vecc{J} \vecc{\gamma}^{*} + \vecc{\gamma} \vecc{J} = \vecc{\sigma} \vecc{\sigma}^{*}. 
\end{equation}
Then the process $\vecc{x}^m_{t}$ converges, as $m \to 0$, to the solution $\vecc{X}_{t}$, of the It\^o SDE \eqref{sklim}, in the following sense:  for all finite $T>0$, 
\begin{equation}
\sup_{t \in [0,T]} |\vecc{x}_t^m - \vecc{X}_t| \to 0
\end{equation} in probability, in the limit as $m \to 0$.  

\end{theorem}


We end this section with a few remarks concerning the statements in Theorem \ref{skthm}.

\begin{remark} \label{ass_rmk} 
\hspace{5cm}
\begin{itemize} 
\item[(i)] Because of the relaxed spectral assumption on $\vecc{\gamma}$, a new idea has to be used to prove decay estimates for solutions of the velocity equation.  Once this is done, Theorem 1 can be proven using the technique of \cite{birrell2017homogenization} (note that Assumption \ref{a1} is essentially the same as the assumptions in Appendix A of \cite{birrell2017homogenization}, specialized to the present case).  In Appendix \ref{proof_sketch} we give a sketch of the proof of Theorem \ref{skthm}, pointing out the necessary modifications.  The reader is referred to \cite{birrell2017homogenization} for more details.
\item[(ii)] Our assumption on the initial variable $\vecc{v}_0^m$ implies that the  initial average kinetic energy, $K(\vecc{v}^m) := \mathbb{E}  \frac{1}{2} m |\vecc{v}^m|^2$, does not blow up (but can possibly vanish) as $m \to 0$. This is analogous to the  Assumption 2.4 in \cite{birrell2017homogenization} and it is  more general than the one in \cite{hottovy2015smoluchowski}.
\item[(iii)] With slightly more work and additional assumptions, one could prove the statement in Assumption \ref{a4} from Assumptions \ref{a1}-\ref{a3}. However, such existence and uniqueness result is not the focus of this paper and, therefore, we choose to take the existence and uniqueness result for granted in Assumption \ref{a4}.
\item[(iv)] We make no claim that Assumptions \ref{a2}-\ref{a4} are as weak or as general as possible.  In particular, the boundedness assumption on the coefficients of the SDEs could be relaxed (for instance, using the techniques in \cite{Herzog2016}) and the initial condition $\vecc{x}$ could have some dependence on $m$ (see, for instance, \cite{birrell2017homogenization}) at the cost of more technicalities. 

The main focus of our revisit here is to point out that the result in \cite{hottovy2015smoluchowski} still holds with a relaxed spectral assumption on the matrix $\vecc{\gamma}$ and with the initial condition $\vecc{v}_0^m$ possibly dependent on $m$ -- this will be important for applications in later sections (see also Remark \ref{imp_rmk}).  
\end{itemize}
\end{remark}

\section{Homogenization for Generalized Langevin Dynamics} \label{general_homog}

In this section, we study homogenization for the system of equations \eqref{general_side_rescaled}-\eqref{general_rescaled_ou}  (with $\mu = \theta = \nu$) by taking the limit as $\epsilon \to 0$, under appropriate assumptions.

Without loss of generality, we set $\mu = \theta = \nu = 1$. We cast \eqref{general_side_rescaled}-\eqref{general_rescaled_ou} as the system of SDEs for the Markov process $(\vecc{x}^\epsilon_{t}, \vecc{v}^\epsilon_{t}, \vecc{z}^\epsilon_{t}, \vecc{y}^\epsilon_{t}, \vecc{\zeta}^\epsilon_{t}, \vecc{\beta}^\epsilon_{t})$ on the state space $\RR^{d}\times \RR^d \times \RR^{d_1} \times \RR^{d_1} \times \RR^{d_2} \times \RR^{d_2}$:
\begin{align}
d\vecc{x}^\epsilon_{t} &= \vecc{v}^\epsilon_{t} dt, \label{sdec1} \\ 
m_{0} \epsilon  d\vecc{v}^\epsilon_{t} &=  - \vecc{g}(\vecc{x}^\epsilon_{t}) \vecc{C}_1 \vecc{y}^\epsilon_{t} dt + \vecc{\sigma}(\vecc{x}^\epsilon_{t}) \vecc{C}_2 \vecc{\beta}^\epsilon_{t} dt +\vecc{F}(\vecc{x}^\epsilon_{t})dt, \\ 
d \vecc{z}^\epsilon_{t} &= \vecc{y}^\epsilon_{t} dt, \\ 
\tau_{\kappa} \epsilon  d\vecc{y}^\epsilon_{t} &= -\vecc{\Gamma}_1 \vecc{y}^\epsilon_{t} dt + \vecc{M}_1 \vecc{C}_1^* \vecc{h}(\vecc{x}^\epsilon_{t}) \vecc{v}^\epsilon_{t} dt, \\ 
d\vecc{\zeta}^\epsilon_{t} &= \vecc{\beta}^\epsilon_{t} dt, \\ 
\tau_{\xi} \epsilon d\vecc{\beta}^\epsilon_{t} &= -\vecc{\Gamma}_2 \vecc{\beta}^\epsilon_{t} dt +  \vecc{\Sigma}_2 d\vecc{W}^{(q_2)}_{t}, \label{sdec6}
\end{align}
where we have defined the auxiliary process 
\begin{equation}
\vecc{y}^\epsilon_{t} := \frac{1}{\tau_{\kappa} \epsilon} \int_{0}^{t} e^{-\frac{\vecc{\Gamma}_1}{\tau_{\kappa} \epsilon}(t-s)} \vecc{M}_1 \vecc{C}_1^* \vecc{h}(\vecc{x}^\epsilon_{s}) \vecc{v}^\epsilon_{s} ds.
\end{equation}
Here, the initial conditions $\vecc{x}^\epsilon_0 = \vecc{x}$, $\vecc{v}^\epsilon_0 = \vecc{v}$, $\vecc{z}^\epsilon_0 = \vecc{z}$ and $\vecc{\zeta}^\epsilon_0 = \vecc{\zeta}$ are random variables.  Note that $\vecc{y}^\epsilon_0 = \vecc{0}$, and $\vecc{\beta}^\epsilon_0$ is a zero mean Gaussian random variable with covariance $\vecc{M}_2/\tau_\xi \epsilon$. 

Let $\vecc{W}^{(q_2)}_{t}$ be an $\RR^{q_2}$-valued Wiener process on the filtered probability space $(\Omega, \mathcal{F}, \mathcal{F}_{t}, \mathbb{P})$ satisfying the usual conditions \cite{karatzas2012Brownian} and $\mathbb{E}$ denotes expectation with respect to $\mathbb{P}$. 



We use the notation introduced in Section \ref{skrevisited} and make the following assumptions.


\begin{ass} \label{ass1} For every $\vecc{x} \in \RR^{d}$, the vector-valued function $\vecc{F}(\vecc{x})$ is continuous, bounded and Lipschitz in $\vecc{x}$, whereas the matrix-valued functions $\vecc{g}(\vecc{x})$, $\vecc{h}(\vecc{x})$, $\vecc{\sigma}(\vecc{x})$, $(\vecc{g})_{\vecc{x}}(\vecc{x})$, $(\vecc{h})_{\vecc{x}}(\vecc{x})$ and $(\vecc{\sigma})_{\vecc{x}}(\vecc{x})$   are continuously differentiable, bounded and Lipschitz in $\vecc{x}$.  Moreover,   $(\vecc{g})_{\vecc{x}\vecc{x}}(\vecc{x})$, $(\vecc{h})_{\vecc{x}\vecc{x}}(\vecc{x})$ and $(\vecc{\sigma})_{\vecc{x}\vecc{x}}(\vecc{x})$ are bounded for every $\vecc{x} \in \RR^d$. 
\end{ass}

\begin{ass} \label{ass2} The initial conditions $\vecc{x}$, $\vecc{v}$, $\vecc{z}$, $\vecc{\zeta}$ are random variables  independent of $\epsilon$.   We assume that they have finite moments of all orders, i.e. $\mathbb{E}|\vecc{x}|^{p}, \ \mathbb{E}|\vecc{v}|^{p}, \  \mathbb{E}|\vecc{z}|^p, \ \mathbb{E}|\vecc{\zeta}|^p < \infty$ for all $p>0$. 
\end{ass}

\begin{ass} \label{ass4} There are no explosions, i.e. almost surely, for any value of the parameter $\epsilon$ there exists a unique solution on the compact time interval $[0,T]$  to the pre-limit equations \eqref{general_side_rescaled}-\eqref{general_rescaled_ou}, and also to the limiting equation \eqref{general_limitSDE}. 
\end{ass}


The following convergence theorem is the main result of this paper. It provides a homogenized SDE for the particle's position in the limit as the inertial time scale, the memory time scale and the noise correlation time scale go to zero at the same rate in the case when the pre-limit dynamics are described by the family of equations \eqref{general_side_rescaled}-\eqref{general_rescaled_ou} (with $\mu = \theta = \nu = 1$), or equivalently by the SDEs \eqref{sdec1}-\eqref{sdec6}.
In the following, $(\vecc{D})_{ij}$ denotes the $(i,j)$-entry of the matrix $\vecc{D}$.

\begin{theorem} \label{general_result}
Let $\vecc{x}^\epsilon_{t} \in \RR^{d}$ be the solution to the SDEs \eqref{sdec1}-\eqref{sdec6}. Suppose that Assumptions \ref{ass1}-\ref{ass4} are satisfied and the effective damping and diffusion (constant) matrices, $\vecc{K}_1$, $\vecc{K}_2$, defined in \eqref{eff_damping} and \eqref{eff_diff} respectively, are invertible. Moreover, we assume that for every $\vecc{x} \in \RR^d$,  
\begin{equation} \label{inv_cond}
\vecc{B}_\lambda(\vecc{x}) := \vecc{I} + \vecc{g}(\vecc{x}) \tilde{\vecc{\kappa}}(\lambda \tau_{\kappa}) \vecc{h}(\vecc{x})/\lambda m_0
\end{equation}
is invertible for all $\lambda$ in the right half plane $\{\lambda \in \CC: Re(\lambda)>0\}$, where $\tilde{\vecc{\kappa}}(z) := \vecc{C}_1(z\vecc{I} + \vecc{\Gamma}_1)^{-1}\vecc{M}_1 \vecc{C}_1^*$, for $z \in \CC$, is the Laplace transform of the memory function. 

Denote $\vecc{\theta}(\vecc{X}) := \vecc{g}(\vecc{X})\vecc{K}_1 \vecc{h}(\vecc{X}) \in \RR^{d \times d}$ for $\vecc{X} \in \RR^d$. Then as $\epsilon \to 0$, the process $\vecc{x}^\epsilon_{t}$ converges to the solution, $\vecc{X}_{t}$, of the following It\^o SDE:
\begin{equation} \label{general_limitSDE}
d\vecc{X}_{t} = \vecc{S}(\vecc{X}_{t}) dt + \vecc{\theta}^{-1}(\vecc{X}_t) \vecc{F}(\vecc{X}_{t}) dt + \vecc{\theta}^{-1}(\vecc{X}_t)  \vecc{\sigma}(\vecc{X}_{t}) \vecc{C}_2 \vecc{\Gamma}_2^{-1} \vecc{\Sigma}_2 d\vecc{W}^{(q_2)}_{t},
\end{equation}
with $\vecc{S}(\vecc{X}_{t}) = \vecc{S}^{(1)}(\vecc{X}_{t}) + \vecc{S}^{(2)}(\vecc{X}_{t}) + \vecc{S}^{(3)}(\vecc{X}_{t}),$ where the $\vecc{S}^{(k)}$ are the noise-induced drifts whose $i$th components are given by 
\begin{align}
S^{(1)}_{i} &= m_{0} \frac{\partial}{\partial X_{l}}\left[(\vecc{\theta}^{-1})_{ij}(\vecc{X})\right] (\vecc{J}_{11})_{jl}(\vecc{X}), \ \
i,j,l = 1, \dots, d, \label{gen_nid1} \\ 
S^{(2)}_{i} &= -\tau_{\kappa} \frac{\partial}{\partial X_{l}}\left[(\vecc{\theta}^{-1} \vecc{g})_{ij}(\vecc{X})\right] (\vecc{C}_1 \vecc{\Gamma}_1^{-1} \vecc{J}_{21})_{jl}(\vecc{X}), \ \ i,l = 1, \dots, d; \ j = 1,\dots,q, \label{gen_nid2} \\ 
S^{(3)}_{i} &= \tau_{\xi} \frac{\partial}{\partial X_{l}}\left[(\vecc{\theta}^{-1}\vecc{\sigma} )_{ij}(\vecc{X}) \right] (\vecc{C}_2 \vecc{\Gamma}_2^{-1} \vecc{J}_{31})_{jl}(\vecc{X}), \ \ i,l = 1, \dots, d; \ j=1,\dots,r. \label{gen_nid3}
\end{align} 
Here 
$\vecc{J}_{11} = \vecc{J}_{11}^* \in \RR^{d\times d}$, $\vecc{J}_{21}=\vecc{J}_{12}^* \in \RR^{d_1 \times d}$ and $\vecc{J}_{31} = \vecc{J}_{13}^* \in \RR^{d_2 \times d}$  satisfy the following system of five matrix equations: 
\begin{align}
\vecc{g} \vecc{C}_1 \vecc{J}_{12}^* + \vecc{J}_{12} \vecc{C}_1^* \vecc{g}^* &= \vecc{\sigma} \vecc{C}_2 \vecc{J}_{13}^* + \vecc{J}_{13} \vecc{C}_2^* \vecc{\sigma}^*, \label{gen_system} \\
m_0 \vecc{J}_{11} \vecc{h}^* \vecc{C}_1 \vecc{M}_1 + \tau_{\kappa} \vecc{\sigma} \vecc{C}_2 \vecc{J}_{23}^* &= \tau_{\kappa} \vecc{g} \vecc{C}_1 \vecc{J}_{22} + m_0 \vecc{J}_{12} \vecc{\Gamma}_1^*,\\
\tau_{\xi} \vecc{g} \vecc{C}_1 \vecc{J}_{23} + m_0 \vecc{J}_{13} \vecc{\Gamma}_2^* &= \vecc{\sigma} \vecc{C}_2 \vecc{M}_2,  \\ 
\vecc{M}_1 \vecc{C}_1^* \vecc{h} \vecc{J}_{12} + \vecc{J}_{12}^* \vecc{h}^* \vecc{C}_1 \vecc{M}_1 &= \vecc{\Gamma}_1 \vecc{J}_{22} + \vecc{J}_{22} \vecc{\Gamma}_1^*, \\
\tau_{\xi} \vecc{M}_1 \vecc{C}_1^* \vecc{h} \vecc{J}_{13} &= \tau_{\xi} \vecc{\Gamma}_1 \vecc{J}_{23} + \tau_{\kappa} \vecc{J}_{23} \vecc{\Gamma}_2^*. 
. \label{gen_end} 
\end{align}
The convergence is obtained in the following sense: for all finite $T>0$, 
\begin{equation}
\sup_{t \in [0,T]}|\vecc{x}_t^\epsilon - \vecc{X}_t| \to 0
\end{equation}
in probability, in the limit as $\epsilon \to 0$.
\end{theorem}

\begin{remark}
Invertibility of the matrices $\vecc{B}_\lambda(\vecc{x})$ (the assumption \eqref{inv_cond}) is a technical condition which will be used in the proof of the theorem.  We are going to verify it in the special cases and applications discussed later (see Corollary \ref{model3} and Corollary \ref{model4}).  In particular, it follows from the stronger spectral condition, namely that
$\vecc{g}(\vecc{x})\tilde{\vecc{\kappa}}(\mu)\vecc{h}(\vecc{x})$ has no spectrum in the right half plane for any $\mu$ with $Re(\mu) > 0$.   See also Remark \ref{role_of_gamma}.
\end{remark}

\begin{proof}
We denote $\vecc{\hat{x}}^\epsilon_{t} := (\vecc{x}^\epsilon_{t}, \vecc{z}^\epsilon_{t}, \vecc{\zeta}^\epsilon_{t}) \in \RR^{d+d_1+d_2}$ and $\vecc{\hat{v}}^\epsilon_{t} := (\vecc{v}^\epsilon_{t}, \vecc{y}^\epsilon_{t}, \vecc{\eta}^\epsilon_{t}) \in \RR^{d+d_1+d_2} $ and rewrite the above SDE system in the form  $\eqref{gsk1}$-$\eqref{gsk2}$:
\begin{align}
d\vecc{\hat{x}}^{\epsilon}_{t} &= \vecc{\hat{v}}^{\epsilon}_{t} dt \label{gen_nsk1}, \\
\epsilon d\vecc{\hat{v}}^\epsilon_{t} &= - \boldsymbol{\hat{\gamma}}(\vecc{x}^\epsilon_{t}) \hat{\vecc{v}}^\epsilon_{t} dt + \vecc{\hat{F}}(\vecc{x}^\epsilon_{t})dt + \vecc{\hat{\sigma}} d\vecc{W}^{(q_2)}_{t}, \label{gen_nsk2}
\end{align}
with \begin{equation}\boldsymbol{\hat{\gamma}}(\vecc{x}^\epsilon_{t}) = \left[ \begin{array}{ccc}
\vecc{0} & \frac{\vecc{g}(\vecc{x}^\epsilon_{t}) \vecc{C}_1}{m_{0}} & -\frac{\vecc{\sigma}(\vecc{x}^\epsilon_{t}) \vecc{C}_2}{m_{0}}   \\
-\frac{\vecc{M}_1 \vecc{C}_1^* \vecc{h}(\vecc{x}^\epsilon_{t}) }{\tau_{\kappa}}  & \frac{\vecc{\Gamma}_1}{\tau_{\kappa}} & \vecc{0}  \\ 
\vecc{0}  & \vecc{0}  & \frac{\vecc{\Gamma}_2}{\tau_{\xi}} \end{array} \right], \ \ \ \vecc{\hat{F}}(\vecc{x}^\epsilon_{t}) = \begin{bmatrix}
         \frac{\vecc{F}(\vecc{x}^\epsilon_{t})}{m_{0}} \\
         \vecc{0}  \\
         \vecc{0}  \\
        \end{bmatrix}, \ \  \vecc{\hat{\sigma}}= 
         \left[ \begin{array}{c}
   \vecc{0}    \\
 \vecc{0}  \\ 
\frac{\vecc{\Sigma}_2}{\tau_{\xi}}    \end{array} \right], \end{equation}
where $\vecc{\hat{\gamma}} \in \RR^{(d+d_1+d_2) \times (d+d_1+d_2)}$ is a 3 by 3 block matrix with each block a matrix of appropriate dimension; $\vecc{\hat{F}} \in \RR^{d+d_1+d_2}$, $\vecc{\hat{\sigma}} \in \RR^{(d+d_1+d_2)\times q_2}$ and the $\vecc{0}$ appearing in $\vecc{\hat{\gamma}}$, $\vecc{\hat{F}}$ and $\vecc{\hat{\sigma}}$  is a zero vector or matrix of appropriate dimension.

We now want to apply Theorem \ref{skthm} (with $m:=\epsilon$, $n:=d+d_1+d_2$, $\vecc{\gamma}$ replaced by $\vecc{\hat{\gamma}}$, $\vecc{F}$ replaced by $\vecc{\hat{F}}$, $\vecc{\sigma}$ replaced by $\vecc{\hat{\sigma}}$, etc.)  to $\eqref{gen_nsk1}$-$\eqref{gen_nsk2}$. 

It is straightforward to see that Assumption \ref{ass1} implies Assumption \ref{a1} and Assumption \ref{ass4} implies Assumption \ref{a4}. 

To verify Assumption \ref{a3}, note again that $\vecc{y}^\epsilon_0 = \vecc{0}$ and so by Assumption \ref{ass2}, we only need to show that for every $p>0$, $\mathbb{E}|\epsilon \vecc{\beta}^\epsilon_0|^p = O(\epsilon^\alpha)$ as $\epsilon \to 0$, for some constant $\alpha \geq p/2$. To show this, we use the fact that for a mean zero Gaussian random variable, $X \in \RR$, with variance $\sigma^2$,
\begin{equation}
\mathbb{E} |X|^p = \sigma^p \frac{2^{p/2} \Gamma\left(\frac{p+1}{2}\right)}{\sqrt{\pi}},
\end{equation}
for every $p>0$, where $\Gamma$ denotes the gamma function \cite{winkelbauer2012moments}. Applying this to $\vecc{\beta}^\epsilon_0$, we obtain, for every $p>0$, $\mathbb{E} |\vecc{\beta}_0^\epsilon|^p = O(1/\epsilon^{p/2})$ as $\epsilon \to 0$  and so $\mathbb{E}|\epsilon \vecc{\beta}_0^\epsilon|^p = O(\epsilon^{p/2})$ as $\epsilon \to 0$.  Therefore, Assumption \ref{a3} is verified. 

It remains to verify Assumption \ref{a2}, i.e. that $\vecc{\hat{\gamma}}$ is positive stable. Note that $\vecc{\Gamma}_2$ is positive stable by assumption and the triangular-block structure of $\vecc{\hat{\gamma}}$ implies that one only needs to verify that the upper left 2 by 2 block matrix of $\vecc{\hat{\gamma}}$:
\begin{equation}
\vecc{L}(\vecc{x}) = \left[ \begin{array}{cc}
 \vecc{0} & \vecc{g}(\vecc{x}) \vecc{C}_1/m_0   \\
-\vecc{M}_1 \vecc{C}_1^* \vecc{g}(\vecc{x})/\tau_{\kappa}  & \vecc{\Gamma}_1/\tau_{\kappa}  \end{array} \right]\end{equation}
 is positive stable, where $\vecc{x} \in \RR^d$. 

We thus need to show that the resolvent set of $-\vecc{L(\vecc{x})}$,  $\rho(-\vecc{L}(\vecc{x})):=\{\lambda \in \CC : (\lambda \vecc{I} + \vecc{L}(\vecc{x}))^{-1} \text{ exists}\}$, contains the right half plane $\{\lambda \in \CC : Re(\lambda)>0\}$ for every $\vecc{x} \in \RR^d$.

Let $\lambda \in \CC$ such that $Re(\lambda) > 0$.  We will use the following formula for blockwise inversion of a block matrix: provided that $\vecc{S}$ and $\vecc{P}-\vecc{Q}\vecc{S}^{-1} \vecc{R}$ are nonsingular, we have
\begin{equation}
\begin{bmatrix}
       \vecc{P} & \vecc{Q}  \\ 
       \vecc{R} & \vecc{S}  \\
     \end{bmatrix}^{-1} = \begin{bmatrix}
       (\vecc{P}-\vecc{Q}\vecc{S}^{-1}\vecc{R})^{-1} & -(\vecc{P}-\vecc{Q}\vecc{S}^{-1}\vecc{R})^{-1} \vecc{Q} \vecc{S}^{-1}  \\ 
       -\vecc{S}^{-1}\vecc{R}(\vecc{P}-\vecc{Q}\vecc{S}^{-1}\vecc{R})^{-1} & \vecc{S}^{-1} + \vecc{S}^{-1}\vecc{R}(\vecc{P}-\vecc{Q}\vecc{S}^{-1}\vecc{R})^{-1} \vecc{Q} \vecc{S}^{-1}\\
     \end{bmatrix} ,\end{equation}
where $\vecc{P}$, $\vecc{Q}$, $\vecc{R}$, $\vecc{S}$ are matrix sub-blocks of arbitrary dimension. 

Since the matrices $\vecc{A}_{\lambda} := \vecc{\Gamma}_1/\tau_{\kappa} + \lambda \vecc{I}$ and $\vecc{B}_{\lambda}(\vecc{x}) := \vecc{I} + \vecc{g}(\vecc{x}) \tilde{\vecc{\kappa}}(\lambda \tau_{\kappa}) \vecc{h}(\vecc{x})/\lambda m_0$ are invertible for all $\lambda$ in the right half plane by assumption, $\lambda \vecc{I} + \vecc{L}(\vecc{x})$ is indeed invertible for every $\vecc{x}$ and in fact using the above formula for the inverse of a block matrix, we have:
\begin{equation}
(\lambda \vecc{I} + \vecc{L}(\vecc{x}))^{-1} = \left[ \begin{array}{cc}
\vecc{B}_{\lambda}^{-1}(\vecc{x})/\lambda & -\vecc{B}_{\lambda}^{-1}(\vecc{x}) \vecc{g}(\vecc{x}) \vecc{C}_1 \vecc{A}_{\lambda}^{-1}/\lambda m_0   \\
\vecc{A}_{\lambda}^{-1} \vecc{M}_1 \vecc{C}_1^* \vecc{h}(\vecc{x}) \vecc{B}_{\lambda}^{-1}(\vecc{x})/\lambda \tau_{\kappa} & \ \vecc{A}_{\lambda}^{-1} (\vecc{I} - \vecc{M}_1 \vecc{C}_1^* \vecc{h}(\vecc{x})\vecc{B}_{\lambda}^{-1}(\vecc{x}) \vecc{g}(\vecc{x}) \vecc{C}_1 \vecc{A}_{\lambda}^{-1}/\lambda m_0 \tau_{\kappa}) \end{array} \right].
\end{equation}


Therefore,  $\vecc{\hat{\gamma}}$ is invertible and one can compute:
\begin{equation}
\boldsymbol{\hat{\gamma}}^{-1} = \left[ \begin{array}{ccc}
m_0 \vecc{\theta}^{-1} & -\tau_{\kappa} \vecc{\theta}^{-1}\vecc{g} \vecc{C}_1 \vecc{\Gamma}_1^{-1} & \tau_{\xi} \vecc{\theta}^{-1} \vecc{\sigma} \vecc{C}_2 \vecc{\Gamma}_2^{-1}   \\
m_{0} \vecc{\Gamma}_1^{-1} \vecc{M}_1 \vecc{C}_1^* \vecc{h} \vecc{\theta}^{-1} & \  \tau_{\kappa} \vecc{\Gamma}_1^{-1}(\vecc{I}- \vecc{M}_1 \vecc{C}_1^* \vecc{h}\vecc{\theta}^{-1} \vecc{g} \vecc{C}_1 \vecc{\Gamma}_1^{-1})  & \  \  \tau_{\xi}\vecc{\Gamma}_1^{-1}\vecc{M}_1\vecc{C}_1^* \vecc{h} \vecc{\theta}^{-1} \vecc{\sigma}\vecc{C}_2 \vecc{\Gamma}_2^{-1} \\ 
\vecc{0}  & \vecc{0}  & \tau_{\xi} \vecc{\Gamma}_2^{-1} \end{array} \right],\end{equation}
where $\vecc{\theta} = \vecc{g} \vecc{K}_1 \vecc{h}$. 
The result follows by applying Theorem \ref{skthm} to the SDE systems \eqref{gen_nsk1}-\eqref{gen_nsk2}. In particular, a rewriting of the resulting Lyapunov equation \eqref{lyp} gives the system of matrix equations \eqref{gen_system}-\eqref{gen_end}.  
\end{proof}

\begin{remark} \label{role_of_gamma}
In the above proof, the condition of invertibility of $\vecc{B}_\lambda(\vecc{x})$ is only used to guarantee the positive stability of the matrix $\hat{\vecc{\gamma}}$.  Therefore, the conclusion of the theorem holds also when the latter can be established in another way.  This can indeed be done in a number of concrete examples.  
\end{remark}

\begin{remark} \label{mark} Our SIDEs belong to a special class of non-Markovian equations, the so-called {\it quasi-Markovian Langevin equations} \cite{eckmann1999non}. For these equations one can introduce a finite number of auxiliary variables in such a way that the evolution of particle's position and velocity, together with these auxiliary variables, is described by a usual SDE system and is thus Markovian.  We remark that such ``Markovianization" procedure works here because the colored noise can be generated by a linear system of SDEs and the memory kernel satisfies a linear system of ordinary differential equations---both with constant coefficients. If, on the other hand, the memory kernel decays as a power, then there is no finite dimensional extension of the space which would make the solution process Markovian \cite{luczka2005non}. 
\end{remark}

The following corollary uses a linear change of variables in a given SIDE, to arrive at an alternative form of the corresponding homogenized SDE of the form \eqref{general_limitSDE}.  

\begin{corollary} \label{class_thm}
For $i=1,2$, let $\vecc{T}_i$ be arbitrary $d_i \times d_i$ constant invertible matrix, where $d_1,d_2$ are positive integers. For $t \ge 0$, denote   $\vecc{\Gamma}'_i=\vecc{T}_i \vecc{\Gamma}_i \vecc{T}^{-1}_i$, $\vecc{M}_i' = \vecc{T}_i \vecc{M}_i \vecc{T}_i^{*}$, $\vecc{C}'_i =\vecc{C}_i \vecc{T}_i^{-1}$, $(\vecc{\beta}^\epsilon_t)'=\vecc{T}_2 \vecc{\beta}^\epsilon_t$, $\vecc{\Sigma}_i' = \vecc{T}_i \vecc{\Sigma}_i$ and consider the  equations:
\begin{align} 
m_0 \epsilon^{\mu} \ddot{\vecc{x}}^\epsilon_{t} &=  \vecc{F}(\vecc{x}^\epsilon_{t}) - \frac{\vecc{g}(\vecc{x}^\epsilon_{t})}{\tau_{\kappa} \epsilon^{\theta}} \int_{0}^{t} \vecc{C}'_1e^{-\frac{\vecc{\Gamma}'_1}{\tau_{\kappa}\epsilon^{\theta}}(t-s)} \vecc{M}'_1 (\vecc{C}'_1)^* \vecc{h}(\vecc{x}^\epsilon_{s}) \dot{\vecc{x}}^\epsilon_{s} ds +  \vecc{\sigma}(\vecc{x}^\epsilon_{t}) \vecc{C}'_2 (\vecc{\beta}^\epsilon_{t})',  \label{class_side_rescaled} \\
\tau_{\xi}\epsilon^\nu d(\vecc{\beta}^\epsilon_t)' &= -\vecc{\Gamma}'_2  (\vecc{\beta}^\epsilon_t)' dt + \vecc{\Sigma}'_2 d\vecc{W}'_t, \label{class_rescaled_ou}
\end{align}
where $\vecc{W}'_t$ is a $q_2$-dimensional Wiener process and the initial condition, $(\vecc{\beta}^\epsilon_0)'$, is normally distributed with zero mean and covariance of $\vecc{M}'_2/(\tau_{\xi} \epsilon^\nu)$. 

Suppose that Assumptions \ref{ass1}-\ref{ass4} are satisfied and the effective damping and diffusion constants, $\vecc{K}'_i = \vecc{C}'_i (\vecc{\Gamma}_i')^{-1} \vecc{M}_i' (\vecc{C}_i')^{*} = \vecc{K}_i$ ($i=1,2$), are invertible. Moreover, we assume that  $\vecc{I} + \vecc{g}(\vecc{x}) \tilde{\vecc{\kappa}'}(\lambda \tau_{\kappa}) \vecc{h}(\vecc{x})/\lambda m_0$ is invertible for all $\lambda$ in the right half plane $\{\lambda \in \CC: Re(\lambda)>0\}$ and $\vecc{x} \in \RR^d$, where $\tilde{\vecc{\kappa}'}(z) := \vecc{C}'_1(z\vecc{I} + \vecc{\Gamma}'_1)^{-1}\vecc{M}'_1 (\vecc{C}'_1)^* = \vecc{\tilde{\kappa}}(z)$ for $z \in \CC$. 

Let $\mu = \theta = \nu$ in \eqref{class_side_rescaled}-\eqref{class_rescaled_ou}. Then as $\epsilon \to 0$, the process $\vecc{x}^\epsilon_t$ converges, in the similar sense as in Theorem \ref{general_result}, to $\vecc{X}_t$, where $\vecc{X}_{t}$ is the solution of the SDE \eqref{general_limitSDE} with the $\vecc{C}_i$, $\vecc{\Gamma}_i$, $\vecc{M}_i$, $\vecc{\Sigma}_i$ replaced by $\vecc{C}'_i$, $\vecc{\Gamma}'_i$, $\vecc{M}'_i$, $\vecc{\Sigma}'_i$ respectively, and the driving Wiener process $\vecc{W}^{(q_2)}_t$ replaced by $\vecc{W}_t'$.
\end{corollary}

Corollary \ref{class_thm} is an easy consequence of Theorem \ref{general_result}. \\

Next, we discuss a particular, but very important, case when a {\it fluctuation-dissipation relation} holds.  This is, for instance, the case when the pre-limit dynamics are (heuristically) derived from Hamiltonian dynamics (see Appendix \ref{appA}). We will further explore similar cases of fluctuation-dissipation relations for the two sub-classes.

\begin{corollary} \label{gen_fdt}
Let $\vecc{x}^\epsilon_{t} \in \RR^{d}$ be the solution to the SDEs \eqref{sdec1}-\eqref{sdec6}. Suppose that the assumptions of  Theorem \ref{general_result} holds. Moreover, we assume that:
\begin{equation} \label{fdt_con1}
\tau_{\kappa} = \tau_{\xi} = \tau,  \ \   \vecc{\sigma} =  \vecc{g},  \ \  \vecc{h} = \vecc{g}^*, \end{equation} where $\tau$ is a positive constant, and 
\begin{equation} \label{fdt_con2}
  \vecc{C}_1 = \vecc{C}_2 := \vecc{C}, \ \ \vecc{\Gamma}_1 = \vecc{\Gamma}_2 := \vecc{\Gamma}, \ \ \vecc{M}_1 = \vecc{M}_2 := \vecc{M}, \ \ \vecc{\Sigma}_1 = \vecc{\Sigma}_2 := \vecc{\Sigma},
\end{equation}
(so that $q=r$ and $d_1 =d_2$). Denote $\vecc{K} := \vecc{C} \vecc{\Gamma}^{-1} \vecc{M} \vecc{C}^*$. 
Then as $\epsilon \to 0$, the process $\vecc{x}^\epsilon_{t}$ converges to the solution, $\vecc{X}_{t}$, of the following It\^o SDE:
\begin{equation} \label{fdtcase_limitSDE}
d\vecc{X}_{t} = \vecc{S}(\vecc{X}_{t}) dt + [\vecc{g}(\vecc{X}_t) \vecc{K} \vecc{g}^*(\vecc{X}_t)]^{-1} \vecc{F}(\vecc{X}_{t}) dt +  [\vecc{g}(\vecc{X}_t) \vecc{K} \vecc{g}^*(\vecc{X}_t)]^{-1} \vecc{g}(\vecc{X}_t) \vecc{C} \vecc{\Gamma}^{-1} \vecc{\Sigma} d\vecc{W}^{(q_2)}_{t},
\end{equation}
where $\vecc{S}(\vecc{X}_t)$ is the noise-induced drift whose $i$th component is given by 
 
 \begin{equation} 
S_i(\vecc{X}) = m_0 \frac{\partial}{\partial X_{l}}\left[ ((\vecc{g} \vecc{K}\vecc{g}^*)^{-1})_{ij}(\vecc{X})\right] (\vecc{J}_{11})_{jl}(\vecc{X}), \ \
i,j,l = 1, \dots, d, \end{equation} 
where $\vecc{J}_{11}$ solves the following system of three matrix equations:
\begin{align}
m_0 \vecc{J}_{11} \vecc{g} \vecc{C} \vecc{M} + \tau \vecc{g} \vecc{C}(\vecc{J}_{23} + \vecc{J}_{23}^*) &= \tau \vecc{g} \vecc{C} \vecc{J}_{22}+\vecc{g} \vecc{C} \vecc{M}, \label{spec1} \\
\vecc{M} \vecc{C}^* \vecc{g}^* \vecc{g} \vecc{C} \vecc{M} (\vecc{\Gamma}^{-1})^* &= \tau \vecc{M} \vecc{C}^* \vecc{g}^* \vecc{g} \vecc{C} \vecc{J}_{23} (\vecc{\Gamma}^{-1})^* + m_0 (\vecc{\Gamma} \vecc{J}_{23} + \vecc{J}_{23} \vecc{\Gamma}^*), \\
\vecc{M} \vecc{C}^* \vecc{g}^*\vecc{g} \vecc{C} \vecc{M} (\vecc{\Gamma}^{-1})^* + \vecc{\Gamma}^{-1} \vecc{M} \vecc{C}^* \vecc{g}^* \vecc{g} \vecc{C} \vecc{M} &= \tau(\vecc{M} \vecc{C}^* \vecc{g}^* \vecc{\Gamma}^{-1} \vecc{J}_{23}^* \vecc{C}^* \vecc{g}^* + \vecc{\Gamma}^{-1} \vecc{J}_{23}^* \vecc{C}^* \vecc{g}^* \vecc{g} \vecc{C} \vecc{M}) \nonumber \\
&\hspace{1cm} + m_0 (\vecc{\Gamma} \vecc{J}_{22} + \vecc{J}_{22} \vecc{\Gamma}^*). \label{spec3}
\end{align}
The convergence is obtained in the same sense as in Theorem \ref{general_result}. 
\end{corollary}

Eqns. \eqref{fdt_con1}-\eqref{fdt_con2} are a form of fluctuation-dissipation relation familiar from non-equilibrium statistical mechanics \cite{toda2012statistical}. As stationary measures of systems satisfying fluctuation-dissipation relations are in equilibrium with respect to the underlying dynamics, this result is relevant for describing equilibrium properties of such systems in the small mass limit. 

\begin{remark} 
Therefore, if the fluctuation-dissipation relation holds, the noise-induced drift in the limiting SDE reduces to {\it a single term} (later we will see how this term simplifies in some special cases).  This result may have interesting implications for nanoscale systems in equilibrium.
We remark that the conditions for the fluctuation-dissipation relation in Corollary \ref{gen_fdt} can be written in other equivalent forms, up to the transformations in \eqref{transf_realize} and multiplication by a constant. 
\end{remark}

\begin{proof}
The above corollary follows from applying Theorem \ref{general_result}. Indeed, by assumptions of the corollary, \eqref{gen_system} simplifies to:
\begin{equation}
\vecc{g} \vecc{C} (\vecc{J}_{12}-\vecc{J}_{13})^* + (\vecc{J}_{12}-\vecc{J}_{13}) (\vecc{g} \vecc{C})^* = \vecc{0}.
\end{equation}
This implies that $\vecc{J}_{12} = \vecc{J}_{13}$ and, therefore, $\vecc{S}^{(2)}$ and $\vecc{S}^{(3)}$ cancel. Rewriting the resulting system of matrix equations in \eqref{gen_system}-\eqref{gen_end} give \eqref{spec1}-\eqref{spec3}. 
\end{proof}

\section{Homogenization for Models of the Two Sub-Classes} \label{homog}

We now return to the two sub-classes of SIDEs \eqref{genle_general} introduced in Section \ref{nmle}. In this section, we study the effective dynamics described by SIDEs \eqref{goal2} and \eqref{goal3} in the limit as $\epsilon \to 0$. By specializing to these two sub-classes, the convergence result of Theorem \ref{general_result}, in particular the expressions in \eqref{general_limitSDE}-\eqref{gen_end}, can be made more explicit under certain assumptions on the matrix-valued coefficients  and therefore the limiting equation obtained may be useful for modeling purposes. 


\subsection{SIDEs Driven by a Markovian Colored Noise}

The following convergence result provides a homogenized SDE for the particle's position in the limit as the inertial time scale, the memory time scale and the noise correlation time scale vanish at the same rate in the case when the pre-limit dynamics are driven by an Ornstein-Uhlenbeck noise.  

\begin{corollary} \label{model3} 
Let $d=d_1=d_2=q_1=q_2=q=r$. We set, in the SDEs \eqref{sdec1}-\eqref{sdec6}:  $\vecc{\beta}^\epsilon_t = \vecc{\eta}^\epsilon_t$, $\tau_{\xi} = \tau_{\eta}$, $\vecc{W}_t^{(q_2)} = \vecc{W}^{(d)}_t := \vecc{W}_t$ and
\begin{equation}
(\vecc{\Gamma}_1, \vecc{M}_1, \vecc{C}_1) = (\vecc{A}, \vecc{A}, \vecc{I}), \ \  (\vecc{\Gamma}_2, \vecc{M}_2, \vecc{C}_2) = (\vecc{A}, \vecc{A}/2, \vecc{I}),\end{equation}
to obtain  SDEs equivalent to equations \eqref{goal2}-\eqref{rescaled-ou} with $\mu = \theta = \nu = 1$. Let $\vecc{x}^\epsilon_{t} \in \RR^{d}$ be the solution to these equations, with the matrices $\vecc{g}(\vecc{x})$ and $\vecc{h}(\vecc{x})$ positive definite for every $\vecc{x} \in \RR^d$. Suppose that Assumptions \ref{ass1}-\ref{ass4} are satisfied and, moreover, that  $\vecc{g}(\vecc{x})$, $\vecc{h}(\vecc{x})$  and the diagonal matrix $\vecc{A}$ are commuting.  

Then as $\epsilon \to 0$, the process $\vecc{x}^\epsilon_{t}$ converges to the solution, $\vecc{X}_{t}$, of the following It\^o SDE:
\begin{equation} \label{limitSDE}
d\vecc{X}_{t} = \vecc{S}(\vecc{X}_{t}) dt + (\vecc{g} \vecc{h})^{-1}(\vecc{X}_{t})\vecc{F}(\vecc{X}_{t}) dt + (\vecc{g}\vecc{h})^{-1}(\vecc{X}_{t}) \vecc{\sigma}(\vecc{X}_{t}) d\vecc{W}_{t},
\end{equation}
with $\vecc{S}(\vecc{X}_{t}) = \vecc{S}^{(1)}(\vecc{X}_{t}) + \vecc{S}^{(2)}(\vecc{X}_{t}) + \vecc{S}^{(3)}(\vecc{X}_{t}),$ where the $\vecc{S}^{(k)}$ are the noise-induced drifts whose $i$th components are given by 
\begin{align}
S^{(1)}_{i}(\vecc{X})  &= m_{0} \frac{\partial}{\partial X_{l}}[((\vecc{g}\vecc{h})^{-1})_{ij}(\vecc{X})] (\vecc{J}_{11})_{jl}(\vecc{X}), \ \
i,j,l = 1, \dots, d, \label{nid1} \\ 
S^{(2)}_{i}(\vecc{X})  &= -\tau_{\kappa} \frac{\partial}{\partial X_{l}}[((\vecc{A} \vecc{h})^{-1})_{ij}(\vecc{X})] (\vecc{J}_{21})_{jl}(\vecc{X}), \ \ i,j,l = 1, \dots, d, \label{nid2} \\ 
S^{(3)}_{i}(\vecc{X}) &= \tau_{\eta} \frac{\partial}{\partial X_{l}}[((\vecc{g}\vecc{h})^{-1} \vecc{\sigma} \vecc{A}^{-1} )_{ij}(\vecc{X})] (\vecc{J}_{31})_{jl}(\vecc{X}), \ \ i,j,l = 1, \dots, d. \label{nid3}
\end{align}
Here
$\vecc{J}_{11} = \vecc{J}_{11}^*$, $\vecc{J}_{21}=\vecc{J}_{12}^*$ and $\vecc{J}_{31} = \vecc{J}_{13}^*$ are $d$ by $d$ block matrices satisfying the following system of matrix equations: 
\begin{align}
\tau_{\eta} \vecc{g}  \vecc{J}_{23} + m_0 \vecc{J}_{13} \vecc{A} &= \vecc{\sigma}  \vecc{A}/2, \label{subclass1_start} \\ 
\tau_{\eta} \vecc{A}  \vecc{h} \vecc{J}_{13} &= \tau_{\eta} \vecc{A} \vecc{J}_{23} + \tau_{\kappa} \vecc{J}_{23} \vecc{A}, \\ 
\vecc{A} \vecc{h} \vecc{J}_{12} + \vecc{J}_{12}^* \vecc{h} \vecc{A} &= \vecc{A} \vecc{J}_{22} + \vecc{J}_{22} \vecc{A}, \\ 
\vecc{g}  \vecc{J}_{12}^* + \vecc{J}_{12}  \vecc{g} &= \vecc{\sigma}\vecc{J}_{13}^* + \vecc{J}_{13}  \vecc{\sigma}^*, \\ 
m_0 \vecc{J}_{11} \vecc{h} \vecc{A} + \tau_{\kappa} \vecc{\sigma} \vecc{J}_{23}^* &= \tau_{\kappa} \vecc{g}  \vecc{J}_{22} + m_0 \vecc{J}_{12} \vecc{A}. \label{subclass1_end} 
\end{align}
The convergence is obtained in the same sense as in Theorem \ref{general_result}.
\end{corollary}
\begin{proof} We will apply Theorem \ref{general_result}. As $\vecc{K}_1=2 \vecc{K}_2 = \vecc{I}$, clearly they are invertible. Also, being positive definite, $\vecc{g}$ and $\vecc{h}$ are invertible. 

Since $\vecc{g}$, $\vecc{h}$ and $\vecc{A}$ are positive definite and commuting matrices, the matrix $\vecc{B}_{\lambda}(\vecc{x})$, defined in \eqref{inv_cond}, is invertible for all $\lambda$ such that $Re(\lambda) > 0$. Indeed, in this case  $\vecc{B}_{\lambda}(\vecc{x}) = \vecc{I} + \vecc{g}(\vecc{x})(\lambda \tau_{\kappa} \vecc{I} + \vecc{A})^{-1} \vecc{A} \vecc{h}(\vecc{x})/\lambda m_0$. Since $\vecc{g}$, $\vecc{h}$ and $\vecc{A}$ are positive definite and commuting, they have positive eigenvalues and can be simultaneously diagonalized. Therefore, all the eigenvalues of $\vecc{B}_{\lambda}(\vecc{x})$ are nonzero for every $\lambda$ with $Re(\lambda) > 0$ and $\vecc{x} \in \RR^d$, so the invertibility condition is verified.  Therefore, the block matrix:
\begin{equation}  \label{nsk2}
\boldsymbol{\hat{\gamma}}(\vecc{x}^\epsilon_{t}) = \left[ \begin{array}{ccc} \vecc{0} & \frac{\vecc{g}(\vecc{x}^\epsilon_{t})}{m_{0}} & -\frac{\vecc{\sigma}(\vecc{x}^\epsilon_{t})}{m_{0}}   \\ -\frac{\vecc{A} \vecc{h}(\vecc{x}^\epsilon_{t}) }{\tau_{\kappa}}  & \frac{\vecc{A}}{\tau_{\kappa}} & \vecc{0}  \\  \vecc{0}  & \vecc{0}  & \frac{\vecc{A}}{\tau_{\eta}} \end{array} \right],
\end{equation}
 is positive stable (see Remark \ref{role_of_gamma}).
The result then follows by applying Theorem \ref{general_result}. 
\end{proof}

For special one-dimensional systems, the form of the limiting equation can be made even more explicit. 
\begin{corollary} \label{1dcase}
In the one-dimensional case, we drop the boldface and write $\vecc{X}_{t} := X_{t} \in \RR, \ \vecc{g}(\vecc{X}) := g(X),$ with $g: \RR \to \RR$, etc.. We assume that $h = g$ and $\vecc{A} := \alpha > 0$ is a constant. The homogenized equation is given by:
\begin{equation} \label{onedlimitSDE}
dX_{t} = S(X_{t}) dt + g^{-2}(X_{t}) F(X_{t}) dt + g^{-2}(X_{t}) \sigma(X_{t}) dW_{t},
\end{equation}
with $S(X_{t}) = S^{(1)}(X_{t}) + S^{(2)}(X_{t}) + S^{(3)}(X_{t}),$ where the noise-induced drift terms $S^{(k)}(X_{t})$ have the following explicit expressions that depend on the parameters $m_{0}, \tau_{\kappa}$ and $\tau_{\eta}$:
\begin{align} 
S^{(1)}(X_{t}) &= \left(\frac{1}{g^2(X_{t})} \right)'  \frac{  \sigma(X_{t})^2}{2 g^2(X_{t})} \left[\frac{\tau_{\kappa}^2 g^2(X_{t})+m_{0}\alpha (\tau_{\kappa}+\tau_{\eta})}{\tau_{\eta}^2 g^2(X_{t})+m_{0}\alpha (\tau_{\kappa}+\tau_{\eta})} \right], \label{86} \\ 
S^{(2)}(X_{t}) &=  - \left(\frac{1}{g(X_{t})} \right)'  \frac{ \sigma(X_{t})^2 \tau_{\kappa}(\tau_{\kappa}+\tau_{\eta})}{2g(X_{t})[\tau_{\eta}^2 g^2(X_{t})+m_{0}\alpha(\tau_{\kappa}+\tau_{\eta})]},  \label{87} \\
S^{(3)}(X_{t}) &= \left(\frac{\sigma(X_{t})}{g^2(X_{t})} \right)'  \frac{ \sigma(X_{t})\tau_{\eta} (\tau_{\kappa}+\tau_{\eta})}{2[\tau_{\eta}^2 g^2(X_{t})+m_{0}\alpha (\tau_{\kappa}+\tau_{\eta})]}, \label{88}  
\end{align}
where the prime $'$ denotes derivative with respect to $X_t$. 
\end{corollary}
\begin{proof}
With $\vecc{x}^\epsilon_{t} := (x^\epsilon_{t}, z^\epsilon_{t}, \zeta^\epsilon_{t}) \in \RR^{3}$ and $\vecc{v}^\epsilon_{t} := (v^\epsilon_{t},  y^\epsilon_{t}, \eta^\epsilon_{t}) \in \RR^{3}$, SDEs \eqref{gen_nsk1}-\eqref{gen_nsk2} become:
\begin{align}
d\vecc{x}^\epsilon_{t} &= \vecc{v}^\epsilon_{t} dt, \\
\epsilon  d\vecc{v}^\epsilon_{t} &= - \boldsymbol{\gamma}(\vecc{x}^\epsilon_{t}) \vecc{v}^\epsilon_{t} dt + \vecc{F}(\vecc{x}^\epsilon_{t})dt + \vecc{\sigma} d\vecc{W}_{t},
\end{align}
where 
\begin{equation} \label{why}
\boldsymbol{\gamma}(\vecc{x}^\epsilon_{t}) = \left[ \begin{array}{ccc}
0 & \frac{g(x^\epsilon_{t})}{m_{0}} & -\frac{\sigma(x^\epsilon_{t})}{m_{0}}   \\
-\frac{\alpha}{\tau_{\kappa}} g(x^\epsilon_{t}) & \frac{\alpha}{\tau_{\kappa}} & 0 \\ 
0 & 0 & \frac{\alpha}{\tau_{\eta}} \end{array} \right], \ \ \vecc{F}(\vecc{x}^\epsilon_{t}) = \begin{bmatrix}
         \frac{F(x^\epsilon_{t})}{m_{0}} \\
         0\\
         0 \\
        \end{bmatrix}, \ \  \vecc{\sigma} = 
         \left[ \begin{array}{c}
0   \\
0  \\ 
\frac{\alpha}{\tau_{\eta}}   \end{array} \right]. \end{equation}
It follows from Corollary \ref{model3} that the matrix $\vecc{\gamma}$ is positive stable; one can also calculate its eigenvalues explicitly and see that their real parts are positive. 
The eigenvalues of $\vecc{\gamma}$ are 
\begin{equation} \frac{\alpha}{\tau_{\eta}}, \ \frac{\alpha}{2 \tau_{\kappa}} \pm \frac{1}{2} \sqrt{\frac{\alpha^2 m_{0}-4 \alpha g(x^\epsilon_{t})^2 \tau_{\kappa}}{m_{0} \tau_{\kappa}^2}},  \end{equation}
and so their real parts are indeed positive.

On the other hand, the solution, $\boldsymbol{J} \in \RR^{3 \times 3}$, to the Lyapunov equation, 
\begin{equation}
\boldsymbol{\gamma} \boldsymbol{J} + \boldsymbol{J} \boldsymbol{\gamma}^{*} = \vecc{\sigma} \vecc{\sigma}^{*},\end{equation}
can be computed (using Mathematica$\textsuperscript{\textregistered}$) to be:
\begin{equation}
\boldsymbol{J} =  \left[ \begin{array}{ccc}
\frac{ \sigma^2}{2m_{0} g^2}\left[ \frac{\tau_{\kappa}^2 g^2 + m_{0} \alpha (\tau_{\kappa}+\tau_{\eta}) }{\tau_{\eta}^2 g^2 + m_{0} \alpha (\tau_{\kappa}+\tau_{\eta})}\right] & \frac{\alpha \sigma^2(\tau_{\kappa}+\tau_{\eta}) }{2g(\tau_{\eta}^2 g^2+m_{0} \alpha(\tau_{\kappa}+\tau_{\eta})} 
& \frac{\alpha \sigma(\tau_{\kappa}+\tau_{\eta})}{2(\tau_{\eta}^2 g^2 + m_{0} \alpha(\tau_{\kappa}+\tau_{\eta}))} \\ 
\frac{\alpha \sigma^2(\tau_{\kappa}+\tau_{\eta}) }{2g(\tau_{\eta}^2 g^2+m_{0} \alpha(\tau_{\kappa}+\tau_{\eta})} 
& \frac{\alpha \sigma^2 (\tau_{\kappa}+\tau_{\eta})}{2(\tau_{\eta}^2 g^2 +m_{0} \alpha(\tau_{\kappa}+\tau_{\eta}))} & \frac{\tau_{\eta} \alpha \sigma g}{2  (\tau_{\eta}^2 g^2+ m_{0} \alpha(\tau_{\kappa}+\tau_{\eta}))} \\ 
 \frac{\alpha \sigma(\tau_{\kappa}+\tau_{\eta})}{2(\tau_{\eta}^2 g^2 + m_{0} \alpha(\tau_{\kappa}+\tau_{\eta}))}
 &  \frac{\tau_{\eta} \alpha \sigma g}{2  (\tau_{\eta}^2 g^2 + m_{0} \alpha(\tau_{\kappa}+\tau_{\eta}))} & \frac{\alpha}{2 \tau_{\eta}}
 \end{array} \right].\end{equation} The result then follows from Corollary \ref{model3}. 
\end{proof}

\begin{remark} \label{imp_rmk}
Note that here the matrix $\vecc{\gamma}$ in \eqref{why} is not symmetric and the smallest eigenvalue of its symmetric part can be negative. Moreover, the initial condition $\vecc{v}^\epsilon_0$ depends on $\epsilon$ through the component $\eta^\epsilon_0$ (which is a zero mean Gaussian random variable with variance $\alpha/2\epsilon$). Thus, we cannot apply the main results in \cite{hottovy2015smoluchowski}  to obtain the convergence result. This is our main motivation to revisit the Smoluchowski-Kramers limit of SDEs in Section III under a weakened spectral assumption on the matrix $\vecc{\gamma}$ (or $\vecc{\hat{\gamma}}$ in the multidimensional case) and a relaxed assumption concerning the $\epsilon$ dependence of $\vecc{v}^\epsilon_0$ (or $\hat{\vecc{v}}^\epsilon_0$ in the multidimensional case).  
\end{remark}

\begin{remark} \label{fdt_impcase} In the important case when the fluctuation-dissipation relation (i.e. $\tau_{\kappa} = \tau_{\eta}$, $h = g$ and $g$ is proportional to  $\sigma$) holds for the one-dimensional models of the first sub-class, the correction drift terms $S^{(2)}$ and $S^{(3)}$ cancel each other and the resulting (single) noise-induced drift term coincides with that obtained in the limit as $m \to 0$ of the systems with no memory, driven by a white noise to which Theorem \ref{skthm} applies directly! However, when the relation fails, we obtain three different drift corrections induced by vanishing of all time scales. Again, the presence of these correction terms may have significant consequences for the dynamics of the systems (see Section \ref{sec:thermophoresis}).
\end{remark}


\subsection{SIDEs Driven by a Non-Markovian Colored Noise}

The following corollary provides a homogenized SDE for the particle position in the limit, in which the inertial time scale, the memory time scale and the noise correlation time scale vanish at the same rate in the case when the pre-limit dynamics are driven by the  harmonic noise.  We emphasize that in this case the original system is driven by a noise which is not a Markov process. 

\begin{corollary} \label{model4}
Let $d=d_1=d_2=q_1=q_2=q=r$.  We set, in the SDEs \eqref{sdec1}-\eqref{sdec6}:   $\tau_{\xi} = \tau_{h}$, $\vecc{W}_t^{(q_2)} = \vecc{W}^{(d)}_t = \vecc{W}_t$ and
\begin{equation}
\vecc{\Gamma}_2 =  
\begin{bmatrix} 
\vecc{0} & -\vecc{I} \\
\vecc{\Omega}^2 & \vecc{\Omega}^2 
\end{bmatrix}, \ \ \ 
\vecc{\Gamma}_1 = \frac{1}{2}
\begin{bmatrix} 
\vecc{\Omega}^2 & 4\vecc{I}-\vecc{\Omega}^2 \\
-\vecc{\Omega}^2 & \vecc{\Omega}^2 
\end{bmatrix} =: \vecc{T} \vecc{\Gamma}_2 \vecc{T}^{-1}, \end{equation}
\begin{equation}
\vecc{M}_2  =  
\frac{1}{2} \begin{bmatrix} 
\vecc{I} & \vecc{0} \\
\vecc{0} & \vecc{\Omega}^2 
\end{bmatrix}, \ \ \ \vecc{M}_1 = 2\vecc{T} \vecc{M}_2 \vecc{T}^*, \end{equation}

\begin{equation} \vecc{C}_2 = [\vecc{I} \ \  \vecc{0}], \ \ \vecc{C}_1 = \vecc{C}_2 \vecc{T}^{-1},  \end{equation}
to obtain SDEs equivalent to equations \eqref{goal3}-\eqref{rescaled_h2} with $\mu = \theta = \nu = 1$.  Let $\vecc{x}^\epsilon_{t} \in \RR^{d}$ be the solution to the SDEs \eqref{sdec1}-\eqref{sdec6}, with the matrices $\vecc{g}(\vecc{x})$ and $\vecc{h}(\vecc{x})$ positive definite for every $\vecc{x} \in \RR^d$. Moreover, $\vecc{g}(\vecc{x})$, $\vecc{h}(\vecc{x})$ and the diagonal matrix $\vecc{\Omega}^2$ are commuting. Suppose that Assumptions \ref{ass1}-\ref{ass4} are satisfied.

Then as $\epsilon \to 0$, the process $\vecc{x}^\epsilon_{t}$ converges to the solution, $\vecc{X}_{t}$, of the following It\^o SDE
\begin{equation} \label{limitSDE}
d\vecc{X}_{t} = \vecc{S}(\vecc{X}_{t}) dt + (\vecc{g}\vecc{h})^{-1}(\vecc{X}_{t})\vecc{F}(\vecc{X}_{t}) dt + (\vecc{g} \vecc{h})^{-1}(\vecc{X}_{t}) \vecc{\sigma}(\vecc{X}_{t}) d\vecc{W}_{t},
\end{equation}
with $\vecc{S}(\vecc{X}_{t}) = \vecc{S}^{(1)}(\vecc{X}_{t}) + \vecc{S}^{(2)}(\vecc{X}_{t}) + \vecc{S}^{(3)}(\vecc{X}_{t}),$ where the $\vecc{S}^{(k)}$ are the noise-induced drift terms whose $i$th components are given by the expressions
\begin{align}
S^{(1)}_{i}(\vecc{X})  &= m_{0} \frac{\partial}{\partial X_{l}}[((\vecc{g}\vecc{h})^{-1})_{ij}(\vecc{X})] (\vecc{J}_{11})_{jl}(\vecc{X}), \label{hnid1} \\ 
S^{(2)}_{i}(\vecc{X})  &= -\tau_{\kappa}\left( \frac{\partial}{\partial X_{l}}[( \vecc{h}^{-1})_{ij}(\vecc{X})]   (\vecc{J}_{21})_{jl}(\vecc{X}) + \frac{\partial}{\partial X_{l}}[( \vecc{h}^{-1} (\vecc{I}-2\vecc{\Omega}^{-2}))_{ij}(\vecc{X})] (\vecc{J}_{31})_{jl}(\vecc{X}) \right), \label{hnid2} \\ 
S^{(3)}_{i}(\vecc{X}) &= \tau_{h} \left(\frac{\partial}{\partial X_{l}}[((\vecc{g}\vecc{h})^{-1} \vecc{\sigma})_{ij}(\vecc{X})] (\vecc{J}_{41})_{jl}(\vecc{X}) +  \frac{\partial}{\partial X_{l}}[((\vecc{g}\vecc{h})^{-1} \vecc{\sigma} \vecc{\Omega}^{-2} )_{ij}(\vecc{X})] (\vecc{J}_{51})_{jl}(\vecc{X}) \right), \ \  \label{hnid3}
\end{align} 
where $i,j,l = 1, \dots, d$.
In the above,
\begin{equation} \vecc{\hat{J}} := \begin{bmatrix}
    \vecc{J}_{11}       & \dots & \vecc{J}_{15} \\
    \vdots       & \ddots & \vdots \\
    \vecc{J}_{51}       & \dots & \vecc{J}_{55} 
\end{bmatrix} \in \RR^{5d \times 5d},  \ \text{ with } \vecc{J}_{kl} \in \RR^{d \times d}, \ \ k,l = 1,\dots,5,\end{equation} is the block matrix solving the Lyapunov equation \begin{equation} \vecc{\hat{J}} \vecc{\hat{\gamma}}^{*} + \vecc{\hat{\gamma}} \vecc{\hat{J}} = \vecc{\hat{\sigma}} \vecc{\hat{\sigma}}^{*}, \end{equation}
where 
\begin{equation} \label{hcase2}
\boldsymbol{\hat{\gamma}} = \left[ \begin{array}{ccccc}
\vecc{0}  & \frac{\vecc{g}(\vecc{X}) }{m_{0}} & \frac{\vecc{g}(\vecc{X})}{m_{0}} & -\frac{\vecc{\sigma}(\vecc{X}) }{m_{0}} & \vecc{0}  \\
-\frac{\vecc{h}(\vecc{X}) }{\tau_{\kappa}} & \frac{\vecc{\Omega}^2}{2 \tau_{\kappa}} & \frac{2 \vecc{I}}{\tau_{\kappa}} - \frac{\vecc{\Omega}^2}{2 \tau_{\kappa}} & \vecc{0}  & \vecc{0}  \\ 
\vecc{0}  & -\frac{\vecc{\Omega}^2}{2 \tau_{\kappa}} & \frac{\vecc{\Omega}^2}{2 \tau_{\kappa}} & \vecc{0}  & \vecc{0}  \\ 
\vecc{0}  & \vecc{0}  & \vecc{0}  & \vecc{0}  & -\frac{1}{\tau_{h}} \vecc{I} \\
\vecc{0}  & \vecc{0}  & \vecc{0}  & \frac{\vecc{\Omega}^2}{\tau_{h}} & \frac{\vecc{\Omega}^2}{\tau_{h}}   \end{array} \right], \ \ \ \ 
 \vecc{\hat{\sigma}} = 
         \left[ \begin{array}{c}
\vecc{0}   \\
\vecc{0}   \\
\vecc{0}  \\
\vecc{0}   \\
\frac{\vecc{\Omega}^2}{\tau_{h}}  \end{array} \right]. \end{equation}
In the above,  $\vecc{\hat{\gamma}} \in \RR^{5d \times 5d}$ is a 5 by 5 block matrix with each block an $\RR^{d \times d}$-valued matrix, $ \vecc{\hat{\sigma}} \in \RR^{5d \times d}$ is a 5 by 1 block matrix  with each block a $\RR^{d \times d}$-valued matrix, $\vecc{I}$ is a $d \times d$ identity matrix,  $\vecc{0}$ in $\vecc{\hat{\gamma}}$ and $\vecc{\hat{\sigma}}$  is a $d \times d$ zero matrix,  and $\vecc{W}$ is a $d$-dimensional Wiener process. 
The convergence is obtained in the same sense as in Theorem \ref{general_result}. 
\end{corollary}

Note that the oscillatory nature of covariance function of the harmonic noise in the pre-limit SIDE makes the noise-induced drift in the resulting limiting SDE more complicated (there are more terms) compared to the case of OU process in the first sub-class. Therefore, we write the system of matrix equations that the $\vecc{J}_{kl}$ satisfy in the form of a matrix Lyapunov equation in Corollary \ref{model4}, without breaking it up into equations for individual blocks. This could of course be done, leading to a (more complicated) analog of \eqref{subclass1_start}-\eqref{subclass1_end}. The proof of Corollary \ref{model4}  is essentially identical to the proof of Corollary \ref{model3}, so we omit it.  

Again, for special one-dimensional systems, we are going to make the result more explicit. 

\begin{corollary} \label{h_1dcase}
In the one-dimensional case,  we drop the boldface and write $\vecc{X}_{t} := X_{t} \in \RR, \ \vecc{g}(\vecc{x}) := g(x),$ with $g: \RR \to \RR$, etc.. We assume that $h=g$ and $\vecc{\Omega} := \Omega $ is a real constant. The homogenized equation is given by:
\begin{equation} \label{onedlimitSDE}
dX_{t} = S(X_{t}) dt + g^{-2}(X_{t}) F(X_{t}) dt + g^{-2}(X_{t}) \sigma(X_{t}) dW_{t},
\end{equation}
with $S(X_{t}) = S^{(1)}(X_{t}) + S^{(2)}(X_{t}) + S^{(3)}(X_{t}),$ where the noise-induced drift terms $S^{(k)}(X_{t})$ have the following explicit expressions (computed using Mathematica$\textsuperscript{\textregistered}$) that depend on the parameters $m_{0}, \tau_{\kappa}$ and $\tau_{h}$:
\begin{align} 
S^{(1)}(X) &= m_{0}\left(\frac{1}{g^2(X)}\right)'J_{11}(X),\\ 
S^{(2)}(X) &= -\tau_{\kappa}\left(\frac{1}{g(X)}\right)' \left(J_{21}(X)+\left(1-\frac{2}{\Omega^2} \right) J_{31}(X) \right), \\
S^{(3)}(X) &= \tau_{h} \left( \frac{\sigma(X)}{g^2(X)}\right)' \left(J_{41}(X)+\frac{1}{\Omega^2} J_{51}(X) \right),   
\end{align}
where  the prime $'$ denotes derivative with respect to $X$ and the $J_{kl}(X)$ are given by:
\begin{align}
J_{11}(X) &= \frac{\sigma^2}{2m_{0} g^2 R(X)} \bigg( g^4 \tau_{\kappa}^4(\tau_{\kappa}^2+\tau_{\kappa} \tau_{h}\Omega^2 + \tau_{h}^2 \Omega^2) + m_{0}^2 \Omega^4 (\tau_{\kappa}+\tau_{h})^2 (\tau_{\kappa}^2 + \tau_{h}^2 \nonumber \\
&\ \ \ \ \ \ + \tau_{\kappa} \tau_{h} (\Omega^2-2)) + m_{0} \Omega^2 g^2 (\tau_{\kappa} + \tau_{h}) [\tau_{h}^4 + \tau_{\kappa}^2\tau_{h}^2 (\Omega^2-2)+\tau_{\kappa}^4(\Omega^2-1) \nonumber \\ 
&\ \ \ \ \ \ + \tau_{\kappa}^3 \tau_{h}(2-3\Omega^2+\Omega^4)] \bigg) \\
J_{21}(X) &= \frac{\sigma^2(\tau_{\kappa} + \tau_{h}) \Omega^2}{4g R(X)} \bigg( m_{0}\Omega^2 (\tau_{\kappa} + \tau_{h})(\tau_{\kappa}^2 + \tau_{h}^2 + \tau_{\kappa} \tau_{h} (\Omega^2-2)) \nonumber \\ 
&\ \ \ + g^2(\tau_{\kappa}^4+\tau_{\kappa}^2 \tau_{h}^2 + \tau_{h}^4+\tau_{\kappa}^3 \tau_{h}(\Omega^2-1))   \bigg),  \\
J_{31}(X) &= -\frac{\sigma^2(\tau_{\kappa} + \tau_{h}) \Omega^2}{4g R(X)} \bigg(-m_{0}\Omega^2 (\tau_{\kappa} + \tau_{h})(\tau_{\kappa}^2 + \tau_{h}^2 + \tau_{\kappa} \tau_{h} (\Omega^2-2))  \nonumber \\ 
&\ \ \ + g^2(\tau_{\kappa}^4+\tau_{\kappa}^2 \tau_{h}^2 - \tau_{h}^4+\tau_{\kappa}^3 \tau_{h}(\Omega^2-1))   \bigg), \\
J_{41}(X) &= \frac{1}{2} \bigg(\sigma \Omega^2 (\tau_{\kappa}+\tau_{h}) [g^2 \tau_{h}^4 + m_{0} \Omega^2 (\tau_{\kappa}+\tau_{h})(\tau_{\kappa}^2+\tau_{h}^2+\tau_{\kappa}\tau_{h}(\Omega^2-2))] \bigg), \\ 
J_{51}(X) &= -\frac{1}{2} \bigg(\sigma \Omega^2 (\tau_{\kappa}+\tau_{h}) [m_{0} \Omega^2 (\tau_{\kappa}+\tau_{h})(\tau_{\kappa}^2+\tau_{h}^2+\tau_{\kappa}\tau_{h}(\Omega^2-2)) \nonumber \\ 
&\ \ \ \ \ -g^2\tau_{\kappa}\tau_{h}^2(\tau_{\kappa}+\tau_{h}(\Omega^2-1))] \bigg),
\end{align}
where $g=g(X)$, $\sigma = \sigma(X)$ and
\begin{align} 
R(X) &= g^4 \tau_{h}^4 (\tau_{\kappa}^2 + \tau_{\kappa} \tau_{h} \Omega^2+\tau_{h}^2 \Omega^2) + m_{0}^2 \Omega^4 (\tau_{\kappa}+\tau_{h})^2(\tau_{\kappa}^2+\tau_{h}^2+\tau_{\kappa} \tau_{h} (\Omega^2-2)) \nonumber \\ 
&\ \ \ \ +g^2m_{0}\tau_{h}^2 \Omega^2[\tau_{h}^3\Omega^2+\tau_{\kappa}^3(\Omega^2-2)+\tau_{\kappa}^2 \tau_{h}\Omega^2(\Omega^2-2)+\tau_{\kappa}\tau_{h}^2(2-2\Omega^2+\Omega^4)]. 
\end{align}
\end{corollary} 

Note that if we send $\Omega \to \infty$ in the expressions for the $S^{(i)}(X)$ $(i=1,2,3)$ above, we recover the corresponding expressions given in Corollary \ref{1dcase} (with $\alpha =1$). This is not surprising, since in this limit the harmonic noise becomes an OU process (with $\alpha = 1)$. 

Moreover, when $\tau_{\kappa} = \tau_{h} = \tau$, the noise-induced drift becomes $S(X) = S^{(1)}(X)+S^{(2)}(X)+S^{(3)}(X),$ where
\begin{align} 
S^{(1)} &= \frac{1}{2}\left(\frac{1}{g^2}\right)'\frac{\sigma^2}{g^2},\\ 
S^{(2)} &= -\frac{2\tau \Omega^2 \sigma^2}{g} \left(\frac{1}{g}\right)' \left(\frac{g^2 \tau + m_{0}\Omega^2(\Omega^2-1)}{4 m_{0}^2 \Omega^6+2 g^2 m_{0} \tau \Omega^4(\Omega^2-1)+g^4 \tau^2 (1+2\Omega^2)} \right), \\
S^{(3)} &=  2 \tau \Omega^2 \sigma \left( \frac{\sigma}{g^2}\right)' \left(\frac{g^2 \tau + m_{0}\Omega^2(\Omega^2-1)}{4 m_{0}^2 \Omega^6+2 g^2 m_{0} \tau \Omega^4(\Omega^2-1)+g^4 \tau^2 (1+2\Omega^2)} \right). 
\end{align}
Again, in the case when the fluctuation-dissipation relation holds we see that the noise-induced drift coincides with that obtained in the limit as $m \to 0$ of the Markovian model in Section III. 


\section{Application to the Study of Thermophoresis} \label{sec:thermophoresis}
\subsection{Introduction}
We revisit the dynamics of a free Brownian particle immersed in a heat bath where a temperature gradient is present. This was previously studied in \cite{Hottovy2012a}. It was found there that the particle experiences a drift in response to the temperature gradient, due to the interplay between the inertial time scale and the noise correlation time scale. Such phenomenon is called {\it thermophoresis}. We refer to \cite{Hottovy2012a,piazza2008thermophoresis} and the references therein for further descriptions of this phenomenon, including references to experiments.

Here, we will study the dynamics of the particle in a non-equilibrium heat bath, where a {\it generalized fluctuation-dissipation relation} holds, in which  both the diffusion coefficient and the temperature of the heat bath vary with the position. In contrast to \cite{Hottovy2012a}, we take into account also the memory time scale (in addition to the inertial time scale and the noise correlation time scale) and model the position of the particle as the solution to a SIDE of the form \eqref{genle_general}.  Unlike the model used in \cite{Hottovy2012a}, the model can be derived heuristically from microscopic dynamics by an argument very similar to that of Appendix \ref{appA}.

For a spherical particle of radius $R$ immersed in a fluid of viscosity $\mu$, which in general is a function of the temperature $T = T(x)$ (and thus depends on $x$ as well), the friction (or damping) coefficient $\gamma$ satisfies the Stokes law \cite{toda2012statistical}:
\begin{equation} \label{stokes}
\gamma(x) = 6 \pi \mu(T) R.
\end{equation}
On the other hand, the damping coefficient $\gamma(x)$ and the noise coefficient $\sigma(x)$ are expressed in terms of the diffusion coefficient $D(x)$ and the temperature $T(x)$ as follows: 
\begin{equation} \label{coeff_fd}
\gamma(x) = \frac{k_{B}T(x)}{D(x)}, \ \ \sigma(x) = \frac{k_{B}T(x) \sqrt{2}}{\sqrt{D(x)}}.\end{equation}


In the following, we study two one-dimensional non-Markovian models of thermophoresis. The first model is driven by a Markovian colored noise and the second model by a non-Markovian one.

\subsection{A Thermophoresis Model with Ornstein-Uhlenbeck Noise}
In this section we model evolution of the position, $x_{t} \in \RR$, of a particle by the following SIDE: 
\begin{equation} \label{thermo}
m \ddot{x}_{t} =  - \sqrt{\gamma(x_{t})} \int_{0}^{t} \alpha e^{-\alpha(t-s)} \sqrt{\gamma(x_{s})} \dot{x}_{s} ds +  \sigma(x_{t}) \eta_{t},
\end{equation}
where $\eta_{t}$ is a stationary process, satisfying the SDE:
\begin{equation} \label{outhermo}
d\eta_{t} = -\alpha \eta_{t} dt + \alpha dW_{t}.
\end{equation}
The above equations are obtained by setting $d=1$, $\vecc{F} = 0$, $\vecc{h} = \vecc{g} = g := \sqrt{\gamma}$, $\vecc{\sigma} = \sigma$ in $\eqref{side2}$ and $\vecc{A} = \alpha$ in $\eqref{ou}$, where $\gamma$ and $\sigma$ are given by \eqref{coeff_fd}. Note that the noise correlation function is proportional to the memory kernel in the SIDE \eqref{thermo}, i.e. 
\begin{equation} E[\eta_{t} \eta_{s}] = \frac{\alpha}{2} e^{-\alpha|t-s|} = \frac{1}{2}\kappa_{1}(t-s), \ s,t \geq 0\end{equation} 
as in \eqref{fdt_sc1}. Together with  \eqref{coeff_fd}, this implies that \eqref{thermo} satisfies the generalized fluctuation-dissipation relation (see the statement of Corollary \ref{gen_fdt} and Remark \ref{fdt_impcase}).  Note also that $g$ is a constant multiple of $\sigma$ if and only if $T$ is position-independent. 


We now consider the effective dynamics of the particle in the limit when all the three characteristic time scales vanish at the same rate. In the following, the prime $'$ denotes derivative with respect to the argument of the function.

\begin{corollary} \label{cor1_thermo}
Let $\epsilon > 0$ be a small parameter and let the particle's position, $x^\epsilon_t \in \RR$ ($t \geq 0$), satisfy the following rescaled version of \eqref{thermo}-\eqref{outhermo}:
\begin{align}
dx^\epsilon_t &= v^\epsilon_t dt, \label{rescaled_m1_thermo0} \\
m_0 \epsilon d v^\epsilon_{t} &= \sigma(x^\epsilon_{t}) \eta^\epsilon_{t} dt - \sqrt{\gamma(x^\epsilon_{t})} \left( \int_{0}^{t} \frac{\alpha}{\tau \epsilon} e^{-\frac{\alpha}{\tau \epsilon}\left(t-s\right)} \sqrt{\gamma(x^\epsilon_{s})} v^\epsilon_{s} ds \right) dt, \label{rescaled_m1_thermo1} \\
\tau \epsilon d\eta^\epsilon_t &= -\alpha \eta^\epsilon_t dt + \alpha dW_t, \label{rescaled_m1_thermo2}
\end{align}
where $m_0$, $\alpha$, $\tau$ are positive constants, and $(W_t)$ is a one-dimensional Wiener process. The initial conditions are random variables $x^\epsilon_0 = x$, $v^\epsilon_0 = v$, independent of $\epsilon$ and (statistically) independent of $(W_t)$, and $\eta^\epsilon_0$ is distributed according to the invariant distribution of the SDE \eqref{rescaled_m1_thermo2}. 

Assume that the assumptions of Corollary \ref{model3} are satisfied (in particular, $\gamma(x) > 0$ for every $x \in \RR$). Then, in the limit as $\epsilon \to 0$,  $x^\epsilon_t$ converges (in the same sense as in Corollary \ref{model3}) to the process $X_{t} \in \RR$, satisfying the SDE:
\begin{equation} \label{limitthermo}
dX_{t} = b_{1}(X_{t}) dt + \sqrt{2 D(X_{t})} dW_{t},
\end{equation}
with the noise-induced drift, $b_{1}(X) = S^{(1)}(X) + S^{(2)}(X) + S^{(3)}(X)$, where  
\begin{align}
S^{(1)}(X) &= D'(X)-\frac{D(X)T'(X)}{T(X)},\\
S^{(2)}(X) &= \left[-\frac{k_{B}T(X)D'(X)}{D(X)}+ k_{B}T'(X) \right] \cdot \left[\frac{ \tau D(X)}{\tau k_{B}T(X)+2m_{0}\alpha D(X)} \right], \\
S^{(3)}(X) &= \left[\frac{k_{B}T(X)D'(X)}{D(X)} \right] \cdot \left[\frac{ \tau D(X)}{\tau k_{B}T(X)+2m_{0}\alpha D(X)} \right].
\end{align}
\end{corollary}
\begin{proof}
The corollary follows from  Corollary \ref{1dcase}. In particular, the expressions for $S^{(1)}$, $S^{(2)}$ and $S^{(3)}$ follow from applying Corollary \ref{1dcase} to the present system (see \eqref{86}-\eqref{88}). 
\end{proof}

We give some remarks and discussions of the contents of Corollary \ref{cor1_thermo} before we end this subsection.   

\begin{remark}
We see that in this case a part of $S^{(2)}$ cancels $S^{(3)}$ and therefore the noise-induced drift simplifies to:
\begin{equation} \label{thermodrift}
b_{1}(X) = D'(X)-\frac{2m_{0}\alpha D^2(X)}{\tau k_{B}T(X)+2m_{0}\alpha D(X)}\frac{T'(X)}{T(X)}.
\end{equation}
Using the Stokes law \eqref{stokes} which gives \begin{equation}D(X) = \frac{k_{B}}{6 \pi R} \frac{T(X)}{\mu(T)}, \end{equation} 
where $\mu(T) = \mu(T(X))$, we have
\begin{equation} \label{thermodrift2}
b_{1}(X) = k_{B} T'(X) \left( \frac{\tau}{2(\alpha m_{0} + 3 \pi R \tau \mu(T))} - \frac{\mu'(T) T(X)}{6 \pi R \mu^2(T)} \right).
 \end{equation}
Equation  \eqref{thermodrift} gives the thermophoretic drift in the limit when the three characteristic time scales vanish. Since it arises in the absence of an external force acting on the particle, it is a ``spurious drift" caused by the presence of the temperature gradient and the state-dependence of the diffusion coefficient. Compared to eqn. (101) in \cite{hottovy2015smoluchowski}, the drift term derived here contains a correction term due to the temperature profile.   
\end{remark}

\noindent {\bf Discussion.} We discuss some physical implications of the thermophoretic drift given in $\eqref{thermodrift}$. As discussed in \cite{Hottovy2012a}, the sign of $b_{1}(X)$ determines the direction in which
the particle is expected to travel. The particle will eventually reach some boundaries, which can be either absorbing or reflecting. We are going to consider the reflecting boundaries case. The position of the particle reaches a steady-state distribution $\rho_{\infty}(X)$ in the limit $t \to \infty$. Assuming that the particle is confined to the interval $(a,b)$, $a<b$, one can compute the stationary density:
\begin{equation} \label{stat_den}
\rho_{\infty}(X) = C \exp{\left(-\int_{a}^{X} \frac{2\alpha}{r \gamma(y) + 2\alpha} \frac{T'(y)}{T(y)} dy \right)},
\end{equation} 
where in terms of the original parameters of the model, $r := \tau/m_{0} > 0$, and $C$ is a normalizing constant. In particular, in absence of temperature gradient ($T'(y) = 0$), the particle is equally likely to be found anywhere in $(a,b)$, whereas when a temperature gradient is present, the distribution of the particle's position is not uniform. In the limit $r \to \infty$, the particle's position is again distributed uniformly on  $(a,b)$. On the other hand, in the limit $r \to 0$ the stationary density is inversely proportional to the temperature, i.e. $\rho_{\infty}(X) = \tilde{C} T(X)^{-1},$ where $\tilde{C}$ is a normalizing constant. Thus, the particle is more likely to be found in the colder region. In the special case when $D(X)$ is proportional to $T(X)$, so that $\gamma$ is independent of $X$, we have \begin{equation}\rho_{\infty}(X) = \tilde{C} T(X)^{-\frac{2\alpha}{2\alpha+r \gamma}},\end{equation}
where $\tilde{C}$ is a normalizing constant, so the particle is more likely to be found in the colder region, with the likelihood decreasing as $r$ increases.

Next, we are going to study the sign of the thermophoretic drift directly using \eqref{thermodrift} (this is in contrast to the approach in \cite{Hottovy2012a}, where $\mu(T)$ is expanded around a fixed temperature). We find that $b_{1}(X) > 0$ if and only if $r > r_c$ and $r_c$ is the critical ratio of $\tau/m_0$, given by:
\begin{equation}
r_c =  \frac{\alpha}{3 \pi R \mu(T)} \left(\frac{\mu'(T) T(X)}{\mu(T) -\mu'(T)T(X)} \right), \end{equation}
where $\mu(T) = \mu(T(X))$ is obtained from the Stokes law. 
For $r = r_c$, the stationary density \eqref{stat_den}  reduces to: \begin{equation}\rho^{c}_{\infty}(X) = C \frac{\mu(T(X))}{T(X)}, \end{equation} where $C$ is a normalizing constant. 
Importantly, note that the drift does not change sign if $T$ is independent of $X$.


Finally, we discuss a special case. When $\mu(T)=\mu_{0} > 0$ is a constant (so that  $\gamma(X)$ is a constant), the thermophoretic drift is given by:
\begin{equation}b_{1}(X) = \frac{k_{B}T'(X)}{6 \pi R \mu_{0}} \left[1 - \frac{\alpha}{\alpha + 3 \pi r R  \mu_{0}} \right].\end{equation} In agreement with the result in \cite{Hottovy2012a}, $b_{1}(X)$ has the same sign as $T'(X)$, leading to 
a flow towards the hotter region. The steady-state density is \begin{equation}\rho_{\infty}(X) = C T(X)^{-\frac{\alpha}{\alpha + 3 \pi r R  \mu_{0}}},\end{equation} where $C$ is a normalizing constant, and the particle is more likely to be found in the colder region for all $r > 0$, even though the  thermophoretic drift actually directs the particle towards the hotter regions.  This effect is in agreement with experiments, and is explained by the presence of reflecting boundary conditions.


\subsection{A Thermophoresis Model with Non-Markovian (Harmonic) Noise}

We repeat the analysis of the previous subsection in the case when the colored noise is a harmonic noise. We set $d=1$, $\vecc{F} = 0$, $\vecc{h} = \vecc{g} = g := \sqrt{\gamma}$, $\vecc{\sigma} = \sigma$, $\vecc{\Omega} = \Omega$, $\vecc{\Omega}_0 = \Omega_0 := \Omega \sqrt{1-\Omega^2/4}$, $\vecc{\Omega}_1 = \Omega_1 := \Omega/\sqrt{1-\Omega^2/4}$ (where $|\Omega|<2$) in the SIDE $\eqref{side3}$ and study the effective dynamics of the resulting system as before. The case where $|\Omega|>2$ can be studied analogously. The following result follows from Corollary \ref{h_1dcase}.

\begin{corollary} \label{cor_thermo_2} Let $\epsilon > 0$ be a small parameter and the particle's position, $x^\epsilon_t \in \RR \ (t \geq 0)$, satisfy the following rescaled SDEs:
\begin{align}
dx^\epsilon_t &= v^\epsilon_t dt, \label{rescaled_thermo1} \\
m_0 \epsilon dv^\epsilon_{t} &= -\frac{\sqrt{\gamma(x^\epsilon_{t})}}{\tau \epsilon} \left( \int_{0}^{t} e^{-\frac{\Omega^2}{2\tau \epsilon}\left(t-s\right)}\left[\cos\left(\frac{\Omega_{0}}{\tau \epsilon}(t-s) \right) + \frac{\Omega_{1}}{2} \sin\left(\frac{\Omega_{0}}{\tau \epsilon}(t-s) \right) \right] \sqrt{\gamma(x^\epsilon_s)} v^\epsilon_{s} ds \right) dt \nonumber \\ 
&\ \ \ \ \ + \sigma(x^\epsilon_{t}) h^\epsilon_{t} dt, \label{rescaled_thermo2} \\
\tau \epsilon dh^\epsilon_t &= u^\epsilon_t dt,    \label{rescaled_thermo3} \\
\tau \epsilon du^\epsilon_t &= -\Omega^2 u^\epsilon_t dt - \Omega^2 h^\epsilon_t dt + \Omega^2 dW_t, \label{rescaled_thermo4}
\end{align}
where $m_0$ and $\tau$ are positive constants, $\Omega$, $\Omega_0$ and $\Omega_1$ are constants defined as before, and $(W_t)$ is a one-dimensional Wiener process. The initial conditions are given by the random variables $x^\epsilon_0 = x$, $v^\epsilon_0 = v$, independent of $\epsilon$, and $(h^\epsilon_0, u^\epsilon_0)$ are distributed according to the invariant measure of the SDEs \eqref{rescaled_thermo3}-\eqref{rescaled_thermo4}. 

Assume that the assumptions in Corollary \ref{model4} are satisfied. Then, in the limit as $\epsilon \to 0$, the process $x^\epsilon_t$ converges (in the same sense as Corollary \ref{model4}) to the process $X_{t} \in \RR$, satisfying the SDE: 
\begin{equation}dX_{t} = b_{2}(X_{t}) dt + \sqrt{2D(X_{t})} dW_{t}, \end{equation}
where the noise-induced drift term is given by:
\begin{align}
b_{2}(X) &= D'(X) \\ 
&\ \ \ - \frac{(4m_{0}^2 \Omega^6 D^2(X)+\tau^2 (k_{B}T(X))^2)D(X)}{4m_{0}^2 \Omega^6 D^2(X)+2k_{B}T(X)m_{0}\tau \Omega^4(\Omega^2-1)D(X)+\tau^2(1+2\Omega^2)(k_{B}T(X))^2}\frac{T'(X)}{T(X)}.
\end{align}
\end{corollary}

We next discuss the contents of Corollary \ref{cor_thermo_2}.
\begin{remark}
Note that $b_{2}(X)$ differs from $b_{1}(X)$ obtained previously and $b_{2}(X) \to b_{1}(X)$, with $\alpha =1$ in the expression for $b_{1}(X)$, in the limit $\Omega \to \infty$.
\end{remark}

\noindent {\bf Discussion.}
In the reflecting boundaries case, the stationary distribution of the particle's position is 
\begin{align}
&\rho_{\infty}(X) \nonumber \\
&= C \exp{\left(-\int_{a}^{X} \frac{D(y)(4\Omega^6D^2(y)+r^2 (k_{B}T(y))^2)}{4 \Omega^6 D^2(y)+2r\Omega^4(\Omega^2-1)D(y)k_{B}T(y)+r^2(1+2\Omega^2)(k_{B}T(y))^2 }\frac{T'(y)}{T(y)} dy \right)}, \end{align}  
where $r := \tau/m_{0} > 0$  and $C$ is a normalizing constant. Similarly to the previous model, in the absence of temperature gradient (i.e. when $T$ is a constant), the particle is equally likely to be found anywhere in $(a,b)$. When a temperature gradient is present, distribution of the particle's position is not uniform. However, in contrast to the previous model, in the limit $r \to \infty$ the particle is not distributed uniformly on $(a,b)$ and in the limit $r \to 0$ the stationary density is no longer inversely proportional to the temperature. Both distributions depend on the diffusion coefficient $D(X)$ as well as on the temperature profile $T(X)$.

We can also study the sign of the thermophoretic drift. In this case there can be up to two critical ratios, 
$r_{c}$, at which $b_{2}(X)$ changes sign, as the equation $b_{2}(X) = 0$ is a quadratic equation in $r$.  In the special case when $\mu(T)=\mu_{0} > 0$ is a constant (and thus so is $\gamma(X)$), the thermophoretic drift is given by:
\begin{equation}b_{2}(X) = \frac{k_{B}T'(X)}{6 \pi R \mu_{0}} \left[1 - \frac{\Omega^6+9\pi^2R^2r^2\mu^{2}_{0}}{\Omega^6+3 \pi R r \Omega^4(\Omega^2-1)\mu_{0} +9\pi^2 R^2 r^2(1+2\Omega^2)\mu^2_{0}}  \right].\end{equation} In contrast to the result in previous model, $b_{2}(X)$ has the same sign as $T'(X)$  provided that \begin{equation}r > \frac{\Omega^2(1-\Omega^2)}{6\pi R \mu_{0}}.\end{equation}  Thus, $b_{2}(X)$ and $T'(X)$ do not share the same sign for all $r>0$, unless $ |\Omega| \geq 1.$ Thus, according to this model, presence of a temperature gradient allows us to tune the parameters $(m_{0}, \tau, \Omega)$ to control the direction which the particle travels.
The steady-state density in this case is \begin{equation}\rho_{\infty}(X) = C T(X)^{-\frac{\Omega^6+9\pi^2R^2r^2\mu^{2}_{0}}{\Omega^6+3 \pi R r \Omega^4(\Omega^2-1)\mu_{0} +9\pi^2 R^2 r^2(1+2\Omega^2)\mu^2_{0}}},\end{equation} where $C$ is a normalizing constant. The particle will be more likely found in the colder region for all $r>0$ if $|\Omega| \geq 1$, whereas this might not be true for all $r>0$ if $|\Omega| < 1$. 


\section{Conclusions and Final Remarks} \label{conc}
We have studied homogenization of a  class of GLEs in the limit when three characteristic time scales, i.e. the inertial time, the characteristic memory time in the damping term, and the correlation time of colored noise driving the equations, vanish at the same rate.  We have derived effective equations, which are simpler in three respects:  
\begin{itemize}
\item[1.]  The velocity variables have been homogenized.  As a result, the number of degrees of freedom is reduced and there are no fast variables left.
\item[2.]  The equations become regular SDEs, since the memory time has been taken to zero.  
\item[3.]  The system is driven by a white noise. 
\end{itemize}

Importantly, {\it noise-induced drifts} are present in the limiting equations, resulting from the dependence of the coefficients of the original model on the state of the system.  We have applied the general results to a study of thermophoretic drift, correcting the formulae obtained in an earlier work \cite{Hottovy2012a}.  In systems, satisfying a fluctuation-dissipation relation, the noise-induced drifts in the limiting SDEs for the particle's position reduce to a single term, and for special cases the limiting SDEs coincide with that of \cite{hottovy2015smoluchowski}. However, in the more general case, new terms appear, absent in the case without memory.  To prove the main theorem, we have employed  the main result of \cite{hottovy2015smoluchowski}, proven here in a different version under a relaxed assumption on the damping matrix and the initial conditions.

Homogenization of other specific non-Markovian models can also be studied using the methods of this paper.  An example is a system with exponentially decaying memory kernel, driven by white noise in the limit as the inertial and memory time scales vanish at the same rate.  In this case the noise-induced drift in the limiting equation will consist of two terms, not three, as in the case studied here.

The colored noises considered in this paper have correlations decaying exponentially (short-range memory). It would be interesting to study cases where the GLE is driven by other colored noises such as fractional Gaussian noises, with covariances decaying as a power,  relevant for modeling anomalous diffusion phenomena in fields ranging from biology to finance \cite{kou2008stochastic}. As mentioned in Section 2, we will explore homogenization for GLEs with vanishing effective damping and diffusion constant in a future work.   


\begin{acknowledgements}
The authors were partially supported by NSF grant DMS-1615045. S. Lim is grateful for the support provided by the Michael Tabor Fellowship from the Program in Applied Mathematics at the University of Arizona during the academic year 2017-2018. The authors learned the method of introducing additional variables to eliminate the memory term from E. Vanden-Eijnden. They would like to thank Maciej Lewenstein for insightful discussion on the GLEs.     
\end{acknowledgements}

\bibliographystyle{spmpsci}      

\bibliography{ref}

\begin{thebibliography}{10}
\providecommand{\url}[1]{{#1}}
\providecommand{\urlprefix}{URL }
\expandafter\ifx\csname urlstyle\endcsname\relax
  \providecommand{\doi}[1]{DOI~\discretionary{}{}{}#1}\else
  \providecommand{\doi}{DOI~\discretionary{}{}{}\begingroup
  \urlstyle{rm}\Url}\fi

\bibitem{hottovy2015smoluchowski}
Hottovy, S., McDaniel, A., Volpe, G., Wehr, J.: The Smoluchowski-Kramers limit
  of stochastic differential equations with arbitrary state-dependent friction.
\newblock Communications in Mathematical Physics \textbf{336}(3), 1259--1283
  (2015)

\bibitem[\protect\astroncite{Kou}{2011}]{kou2008stochastic}
\newblock  Stochastic modeling in nanoscale biophysics: Subdiffusion within proteins.
\newblock {\em The Annals of Applied Statistics}, 2(2):501-535 
\newblock Institute of Mathematical Statistics.


\bibitem{majda2001mathematical}
Majda, A.J., Timofeyev, I., Vanden~Eijnden, E.: A mathematical framework for
  stochastic climate models.
\newblock Communications on Pure and Applied Mathematics \textbf{54}(8),
  891--974 (2001)

\bibitem{givon2004extracting}
Givon, D., Kupferman, R., Stuart, A.: Extracting macroscopic dynamics: model
  problems and algorithms.
\newblock Nonlinearity \textbf{17}(6), R55 (2004)

\bibitem{Pavliotis-TwoFast}
Pavliotis, G.A., Stuart, A.M.: Analysis of white noise limits for stochastic
  systems with two fast relaxation times.
\newblock Multiscale Modeling \& Simulation \textbf{4}(1), 1--35 (2005)

\bibitem{Pavliotis}
Pavliotis, G., Stuart, A.: Multiscale Methods, \emph{Texts in Applied
  Mathematics}, vol.~53.
\newblock Springer, New York (2008)

\bibitem{nelson1967dynamical}
Nelson, E.: Dynamical Theories of Brownian
  Motion, Vol.~2.
\newblock Princeton University Press (1967)

\bibitem{franosch2011resonances}
Franosch, T., Grimm, M., Belushkin, M., Mor, F.M., Foffi, G., Forr{\'o}, L.,
  Jeney, S.: Resonances arising from hydrodynamic memory in Brownian Motion.
\newblock Nature \textbf{478}(7367), 85--88 (2011)

\bibitem{Groblacher2015}
Gr\"{o}blacher, S., Trubarov, A., Prigge, N., Cole, G., Aspelmeyer, M., Eisert,
  J.: Observation of non-Markovian micromechanical Brownian motion.
\newblock Nature Communications \textbf{6}, 7606 (2015).
\newblock \doi{10.1038/ncomms8606}.


\bibitem{trentelman2002control}
Trentelman, H.L., Stoorvogel, A.A., Hautus, M.: Control Theory for Linear
  Systems.
\newblock Communications and Control Engineering Series, Springer  (2002)

\bibitem{willems1980stochastic}
Willems, J., Van~Schuppen, J.: Stochastic systems and the problem of state
  space realization.
\newblock In: Geometrical Methods for the Theory of Linear Systems: Proceedings
  of a NATO Advanced Study Institute and AMS Summer Seminar in Applied
  Mathematics held at Harvard University, Cambridge, Mass., June 18--29, 1979,
  vol.~62, p. 283. Springer (1980)

\bibitem{mori1965transport}
Mori, H.: Transport, collective motion, and Brownian motion.
\newblock Progress of Theoretical Physics \textbf{33}(3), 423--455 (1965)

\bibitem{Kubo_fd}
Kubo, R.: The fluctuation-dissipation theorem.
\newblock Reports on Progress in Physics \textbf{29}(1), 255 (1966).


\bibitem{toda2012statistical}
Toda, M., Kubo, R., Saito, N., Hashitsume, 
  N.: Statistical Physics II: Nonequilibrium Statistical Mechanics.
\newblock Springer Series in Solid-State Sciences. Springer Berlin Heidelberg
  (2012).


\bibitem{goychuk2012viscoelastic}
Goychuk, I.: Viscoelastic subdiffusion: Generalized Langevin equation approach.
\newblock Advances in Chemical Physics \textbf{150}, 187 (2012)

\bibitem{van1998remarks}
Van~Kampen, N.: Remarks on non-Markov processes.
\newblock Brazilian Journal of Physics \textbf{28}(2), 90--96 (1998)

\bibitem{luczka2005non}
{\L}uczka, J.: Non-Markovian stochastic processes: Colored noise.
\newblock Chaos: An Interdisciplinary Journal of Nonlinear Science
  \textbf{15}(2), 026,107 (2005)

\bibitem{samorodnitsky1994stable}
Samorodnitsky, G., Taqqu, M.: Stable Non-Gaussian Random Processes: Stochastic
  Models with Infinite Variance.
\newblock Stochastic Modeling Series. Taylor \& Francis (1994).


\bibitem{PhysRevB.89.134303}
Stella, L., Lorenz, C.D., Kantorovich, L.: Generalized Langevin equation: An
  efficient approach to nonequilibrium molecular dynamics of open systems.
\newblock Phys. Rev. B \textbf{89}, 134,303 (2014).
\newblock \doi{10.1103/PhysRevB.89.134303}.


\bibitem{mckinley2009transient}
McKinley, S.A., Yao, L., Forest, M.G.: Transient anomalous diffusion of tracer
  particles in soft matter.
\newblock Journal of Rheology (1978-present) \textbf{53}(6), 1487--1506 (2009)

\bibitem{adelman1976generalized}
Adelman, S., Doll, J.: Generalized Langevin equation approach for
  atom/solid-surface scattering: General formulation for classical scattering
  off harmonic solids.
\newblock The Journal of Chemical Physics \textbf{64}(6), 2375--2388 (1976)

\bibitem{Ottobre}
{Ottobre}, M., {Pavliotis}, G.A.: {Asymptotic analysis for the generalized
  Langevin equation}.
\newblock Nonlinearity \textbf{24}, 1629--1653 (2011).
\newblock \doi{10.1088/0951-7715/24/5/013}

\bibitem{kalman1960new}
Kalman, R.E., et~al.: A new approach to linear filtering and prediction
  problems.
\newblock Journal of Basic Engineering \textbf{82}(1), 35--45 (1960)

\bibitem{bellman1997introduction}
Bellman, R.: Introduction to Matrix Analysis, Vol.~19.
\newblock SIAM (1997)

\bibitem{lindquist1985realization}
Lindquist, A., Picci, G.: Realization theory for multivariate stationary
  Gaussian processes.
\newblock SIAM Journal on Control and Optimization \textbf{23}(6), 809--857
  (1985)

\bibitem{lindquist2015linear}
Lindquist, A., Picci, G.: Linear Stochastic Systems: A Geometric Approach to
  Modeling, Estimation and Identification.
\newblock Series in Contemporary Mathematics. Springer Berlin Heidelberg
  (2015).


\bibitem{bao2005non}
Bao, J.D., H{\"a}nggi, P., Zhuo, Y.Z.: Non-Markovian Brownian dynamics and
  nonergodicity.
\newblock Physical Review E \textbf{72}(6), 061,107 (2005)

\bibitem{pavliotis2014stochastic}
Pavliotis, G.: Stochastic Processes and Applications: Diffusion Processes, the
  Fokker-Planck and Langevin Equations.
\newblock Texts in Applied Mathematics. Springer New York (2014).


\bibitem{hottovy2015small}
Hottovy, S., McDaniel, A., Wehr, J.: A small delay and correlation time limit
  of stochastic differential delay equations with state-dependent colored
  noise.
\newblock Markov Processes And Related Fields \textbf{v.22}, Issue 3, 595-628 (2016)

\bibitem{schimansky1990harmonic}
Schimansky-Geier, L., Z{\"u}licke, C.: Harmonic noise: Effect on bistable
  systems.
\newblock Zeitschrift f{\"u}r Physik B Condensed Matter \textbf{79}(3),
  451--460 (1990)

\bibitem{McDaniel14}
{McDaniel}, A., {Duman}, O., {Volpe}, G., {Wehr}, J.: {An SDE approximation for
  stochastic differential delay equations with colored state-dependent noise}.
\newblock arXiv preprint arXiv:1406.7287  (2014)

\bibitem{hanggi1993can}
H{\"a}nggi, P., Jung, P., Zerbe, C., Moss, F.: Can colored noise improve
  stochastic resonance?
\newblock Journal of Statistical Physics \textbf{70}(1), 25--47 (1993)

\bibitem{zwanzig2001nonequilibrium}
Zwanzig, R.: Nonequilibrium Statistical Mechanics.
\newblock Oxford University Press (2001).


\bibitem{karatzas2012Brownian}
Karatzas, I., Shreve, S.: Brownian Motion and Stochastic Calculus, Vol. 113.
\newblock Springer Science \& Business Media (2012)

\bibitem{Hottovy12}
Hottovy, S., Volpe, G., Wehr, J.: Noise-induced drift in stochastic
  differential equations with arbitrary friction and diffusion in the
  Smoluchowski-Kramers limit.
\newblock Journal of Statistical Physics \textbf{146}(4), 762--773 (2012).
\newblock \doi{10.1007/s10955-012-0418-9}.


\bibitem{Herzog2016}
Herzog, D.P., Hottovy, S., Volpe, G.: The small mass limit for Langevin
  dynamics with unbounded coefficients and positive friction.
\newblock Journal of Statistical Physics \textbf{163}(3), 659--673 (2016).
\newblock \doi{10.1007/s10955-016-1498-8}.


\bibitem{birrell2017small}
Birrell, J., Hottovy, S., Volpe, G., Wehr, J.: Small mass limit of a Langevin
  equation on a manifold.
\newblock In: Annales Henri Poincar{\'e}, vol.~18, pp. 707--755. Springer
  (2017)

\bibitem{birrell2017homogenization}
Birrell, J., Wehr, J.: Homogenization of dissipative, noisy, Hamiltonian
  dynamics.
\newblock Stochastic Processes and their Applications  (2017)

\bibitem{volpe2016effective}
Volpe, G., Wehr, J.: Effective drifts in dynamical systems with multiplicative
  noise: A review of recent progress.
\newblock Reports on Progress in Physics \textbf{79}(5), 053,901 (2016)

\bibitem{horn1994topics}
Horn, R., Johnson, C.: Topics in Matrix Analysis.
\newblock Cambridge University Press (1994).


\bibitem{winkelbauer2012moments}
Winkelbauer, A.: Moments and absolute moments of the normal distribution.
\newblock arXiv preprint arXiv:1209.4340  (2012)

\bibitem{eckmann1999non}
Eckmann, J.P., Pillet, C.A., Rey-Bellet, L.: Non-equilibrium statistical
  mechanics of anharmonic chains coupled to two heat baths at different
  temperatures.
\newblock Communications in Mathematical Physics \textbf{201}(3), 657--697
  (1999)

\bibitem{Hottovy2012a}
Hottovy, S., Volpe, G., Wehr, J.: Thermophoresis of Brownian particles driven
  by colored noise.
\newblock EPL \textbf{99}, 60,002 (2012)

\bibitem{piazza2008thermophoresis}
Piazza, R., Parola, A.: Thermophoresis in colloidal suspensions.
\newblock Journal of Physics: Condensed Matter \textbf{20}(15), 153,102 (2008)

\bibitem{ford1965statistical}
Ford, G., Kac, M., Mazur, P.: Statistical mechanics of assemblies of coupled
  oscillators.
\newblock Journal of Mathematical Physics \textbf{6}(4), 504--515 (1965)

\bibitem{hanggi1997generalized}
H{\"a}nggi, P.: Generalized Langevin equations: A useful tool for the perplexed
  modeller of nonequilibrium fluctuations?
\newblock In: Stochastic dynamics, pp. 15--22. Springer (1997)

\bibitem{Zwanzig1973}
Zwanzig, R.: Nonlinear generalized Langevin equations.
\newblock Journal of Statistical Physics \textbf{9}(3), 215--220 (1973).
\newblock \doi{10.1007/BF01008729}.


\bibitem{ariel2008strong}
Ariel, G., Vanden-Eijnden, E.: A strong limit theorem in the Kac--Zwanzig
  model.
\newblock Nonlinearity \textbf{22}(1), 145 (2008)

\bibitem{rey2006open}
Rey-Bellet, L.: Open classical systems.
\newblock In: Open Quantum Systems II, pp. 41--78. Springer (2006)

\bibitem{kabanov2013two}
Kabanov, Y., Pergamenshchikov, S.: Two-Scale Stochastic Systems: Asymptotic
  Analysis and Control.
\newblock Stochastic Modelling and Applied Probability. Springer Berlin
  Heidelberg (2013).


\bibitem{silvester2000determinants}
Silvester, J.R.: Determinants of block matrices.
\newblock The Mathematical Gazette \textbf{84}(501), 460--467 (2000)

\bibitem{allen2007introduction}
Allen, L.: An Introduction to Mathematical Biology.
\newblock Pearson/Prentice Hall (2007).

\bibitem{williams1991probability}
Williams, D.: Probability with Martingales.
\newblock Cambridge University Press (1991).


\end{thebibliography}

\section*{Appendices}
\appendix
\section{Derivation of SIDEs From a Hamiltonian Model} \label{appA}

For completeness, we provide a derivation for a special case of SIDEs \eqref{genle_general}, \eqref{side2} and \eqref{side3} from a Hamiltonian model of a small system (Brownian particle) in contact with a heat bath in  thermal equilibrium.
The particle is moving in a potential $U$. The heat bath is modeled as a system of non-interacting harmonic oscillators whose initial energy is distributed according to the Gibbs distribution at temperature $T$. The Brownian particle is coupled to each oscillator in the bath. This model is used widely to study many systems in statistical physics \citep{ford1965statistical, mori1965transport}. Our goal  is to derive, heuristically, a stochastic integro-differential equation (SIDE) for the position and momentum variables of the particle from the Hamiltonian dynamics. This derivation serves to motivate the class of SIDEs that we are studying in this paper. We emphasize that our derivation here is certainly not original and follows closely that in \citep{hanggi1997generalized} (see also an abstract approach in \citep{Zwanzig1973}).  

One approach to derive the equations is to assume first that there are finitely many harmonic oscillators in the bath (Kac-Zwanzig model \citep{zwanzig2001nonequilibrium,ariel2008strong}). We then take the thermodynamic limit by sending the number of oscillators to infinity in the resulting equations (replacing finite sum over oscillator frequencies by an integral), arguing that the set of frequencies must be dense to allow dissipation of energy from the system to the bath and to eliminate Poincar\'e recurrence. Another approach, which is more technically involved, is to replace the finite system of oscillator equations by a system modeled by a wave equation \citep{rey2006open,pavliotis2014stochastic}. We will derive the SIDEs by adopting the former approach in the multi-dimensional case. 

We consider the situation where the coupling is nonlinear in the particle's position  and linear in the bath variables. Let $\vecc{\hat{x}} = (\vecc{x}, \vecc{x}_{1}, \dots, \vecc{x}_{N}) \in \RR^{d+N p}$ and $\vecc{\hat{p}} = (\vecc{p}, \vecc{p}_{1}, \dots, \vecc{p}_{N}) \in \RR^{d+Np}$ (here, $\vecc{x}, \vecc{p} \in \RR^d$ and $\vecc{x}_i, \vecc{p}_i \in \RR^p$ for $i=1,\dots,N$). Hereafter, the superscript $^*$ denotes transposition and $|\vecc{b}|^2 := \vecc{b}^* \vecc{b} = \sum_{k=1}^{n} b_k^2$ denotes square of the norm of vector $\vecc{b} := (b_1, \dots, b_n) \in \RR^n$. 

The Hamiltonian of the system plus bath is:
\begin{equation}H(\vecc{\hat{x}},\vecc{\hat{p}}) = \frac{|\vecc{p}|^2}{2m} + U(\vecc{x}) + \sum_{k=1}^{N} \left( \frac{|\vecc{p}_{k}|^2}{2} + \frac{1}{2} \omega_{k}^2 \left| \vecc{x}_{k}-\frac{\vecc{c}^*_{k}}{\omega_{k}^2}\vecc{f}(\vecc{x}) \right|^2 \right) , \end{equation}
where $m$ is the particle's mass, $\vecc{x} \in \RR^{d}$ and $\vecc{p} \in \RR^{d}$ are respectively the particle's position  and  momentum, $\vecc{x}_{k} \in \RR^{p}$, $\vecc{p}_{k} \in \RR^{p}$  and $\omega_{k} \in \RR^{+}$ $(k=1,\dots,N)$ are respectively the position, momentum and frequency of the $k$th bath oscillator (with unit mass), $\vecc{f}(\vecc{x}) := (f_{1}(\vecc{x}), \dots, f_{r}(\vecc{x})) \in \RR^r$ is a vector function of $\vecc{x} := (x^{(1)}, \dots, x^{(d)})$ and the $\vecc{c}_{k} \in \RR^{r \times p}$ (so $\vecc{c}_k^* \in \RR^{p \times r}$) are  coupling matrices that specify the coupling strength between the system and the $k$th bath oscillator. 

To derive an equation for the particle's position and momentum, we start by substituting the expression for $H(\hat{\vecc{\vecc{x}}}, \hat{\vecc{\vecc{p}}})$ into the Hamilton's equations to obtain: 
\begin{align}
\dot{\vecc{x}} &= \frac{\vecc{p}}{m}, \\ 
\dot{\vecc{p}} &= -\vecc{\nabla}_{\vecc{x}} U(\vecc{x})  + \vecc{g}(\vecc{x}) \sum_{k} \vecc{c}_{k} \left(\vecc{x}_{k}-\frac{\vecc{c}^*_{k}}{ \omega_{k}^2} \vecc{f}(\vecc{x}) \right), \label{momeq}\\
\vecc{x}_{k} &= \vecc{p}_{k}, \ \ \ k = 1, \dots, N,\\
\vecc{p}_{k} &= -\omega_{k}^2 \vecc{x}_{k} + \vecc{c}^*_{k} \vecc{f}(\vecc{x}), \ \ \ k = 1,\dots,N,
\end{align}
where $\vecc{g}(\vecc{x}) \in \RR^{d \times r}$ denotes the Jacobian matrix $\left(\frac{\partial f_{i}}{\partial x^{(j)}} \right)_{ij}.$

Next, we eliminate the bath variables $\vecc{x}_{k}, \vecc{p}_{k}$, $k=1,\dots,N$, from the system's dynamics. Solving for $\vecc{x}_{k}(t)$ in terms of $\vecc{x}(t)$:
\begin{equation}\vecc{x}_{k}(t) = \vecc{x}_{k}(0)\cos(\omega_{k}t) + \frac{\vecc{p}_{k}(0)}{\omega_{k}} \sin(\omega_{k}t) + \frac{\vecc{c}^*_{k}}{ \omega_{k}} \int_{0}^{t} \sin(\omega_{k}(t-s)) \vecc{f}(\vecc{x}(s)) ds. \end{equation}

Substituting this into $\eqref{momeq}$, we obtain:
\begin{align}
\dot{\vecc{p}}(t) &= -\vecc{\nabla}_{\vecc{x}} U(\vecc{x}(t))  + \vecc{g}(\vecc{x}(t)) \sum_{k} \frac{\vecc{c}_{k} \vecc{c}_{k}^{*}}{ \omega_{k}^2} \left( \int_{0}^{t}  \omega_{k} \sin(\omega_{k}(t-s)) \vecc{f}(\vecc{x}(s)) ds  - \vecc{f}(\vecc{x}(t)) \right) \nonumber \\ 
&\ \ \ \  \ + \vecc{g}(\vecc{x}(t)) \vecc{F}(t), \end{align}
where \begin{equation}\vecc{F}(t) = \sum_{k} \vecc{c}_{k} \left(\vecc{x}_{k}(0) \cos(\vecc{w}_{k}t) + \frac{\vecc{p}_{k}(0)}{ \omega_{k}} \sin(\omega_{k}t) \right).\end{equation}

In the integral term above, we integrate by parts to obtain:
\begin{equation} \int_{0}^{t}  \omega_{k} \sin(\omega_{k}(t-s)) \vecc{f}(\vecc{x}(s)) ds = \vecc{f}(\vecc{x}(t)) - \cos(\omega_{k}t) \vecc{f}(\vecc{x}(0)) - \int_{0}^{t} \cos(\omega_{k}(t-s)) \vecc{g}^*(\vecc{x}(s)) \dot{\vecc{x}}(s) ds.\end{equation}

Using this, the equation for $\vecc{p}(t)$ becomes the generalized Langevin equation (GLE):
\begin{equation} \label{deri_side}
\dot{\vecc{p}}(t) = -\vecc{\nabla}_{\vecc{x}} U(\vecc{x}(t))  - \vecc{g}(\vecc{x}(t))  \int_{0}^{t}  \vecc{\kappa}(t-s) \vecc{g}^*(\vecc{x}(s)) \dot{\vecc{x}}(s) ds    +  \vecc{g}(\vecc{x}(t)) \vecc{\xi}(t),\end{equation}
where 
\begin{equation} \label{ham_memory}
\vecc{\kappa}(t) = \sum_{k} \frac{\vecc{c}_{k} \vecc{c}_{k}^{*}}{ \omega_{k}^2} \cos(\omega_{k}t) \in \RR^{r \times r} \end{equation} 
and 
\begin{equation}\vecc{\xi}(t) = \vecc{F}(t) - \vecc{\kappa}(t) \vecc{f}(\vecc{x}(0)) =  \sum_{k} \vecc{c}_{k} \left( \left( \vecc{x}_{k}(0) - \frac{\vecc{c}^*_{k}}{\omega_{k}^2} \vecc{f}(\vecc{x}(0)) \right) \cos(\omega_{k} t) + \frac{\vecc{p}_{k}(0)}{ \omega_{k}} \sin(\omega_{k}t) \right). \end{equation}

Note that $\vecc{\xi}(t) \in \RR^r$ is expressed in terms of the initial values of the variables $\vecc{x}'_{k}(0) := \vecc{x}_{k}(0) - \frac{\vecc{c}^*_{k}}{\omega_{k}^2} \vecc{f}(\vecc{x}(0)) \in \RR^{p} $ and $\vecc{p}_{k}(0)\in \RR^{p}$. If all these initial values are known, then $\vecc{\xi}(t)$ is a deterministic force. However,  one rarely has a complete information about these initial values and this is where the introduction of randomness can help to simplify the model. In view of this, we assume that the variables $\vecc{x}'_k(0)$ and $\vecc{p}_k(0)$ are random and are distributed according to a Gibbs measure, with the density:
\begin{equation}\rho((\vecc{x}_{k}, \vecc{p}_{k}) \ | \ \vecc{x}(0) = \vecc{x})  = Z^{-1} \exp{\left(-\beta \left(\sum_{k=1}^{N} \frac{|\vecc{p}_{k}|^2}{2} + \frac{1}{2} \omega_{k}^2 \left| \vecc{x}_{k}-\frac{\vecc{c}^*_{k}}{\omega_{k}^2}\vecc{f}(\vecc{x}) \right|^2 \right) \right)},\end{equation} where $\beta = 1/(k_{B}T)$ and $Z$ is the partition function. Taking the averages of the bath variables with respect to the above density:
\begin{align}
&E_{\rho}\left[ \vecc{x}'_{k}(0)  \ |  \ \vecc{x}(0) = \vecc{x} \right] = 0, \ \ \  \ E_{\rho}[ \vecc{p}_{k}(0) \ | \ \vecc{x}(0) = \vecc{x} ] = 0, \\
&E_{\rho} [ \vecc{x}'_k(0) (\vecc{x}'_{k}(0))^* \ | \ \vecc{x}(0) = \vecc{x} ] = \frac{k_{B}T}{\omega_{k}^2} \vecc{I}, \ \  \ \ \ E_{\rho}[ (\vecc{p}_{k}(0) (\vecc{p}_k(0))^* \ | \ \vecc{x}(0) = \vecc{x} ] = k_{B}T \vecc{I}, 
\end{align}
where $E_{\rho}$ denotes mathematical expectation with respect to $\rho$ and $\vecc{I} \in \RR^{p \times p}$ is identity matrix. 

Note that $\vecc{\xi}(t)$ is a stationary, Gaussian process, if it is conditionally averaged with respect to $\rho$ \cite{zwanzig2001nonequilibrium}. It follows from this distribution of the bath variables that we have the  {\it fluctuation-dissipation relation}:
\begin{equation}E_{\rho}[\vecc{\xi}(t)] = 0, \ \ E_{\rho}[\vecc{\xi}(t)\vecc{\xi}(s)^{*}]=k_{B}T\vecc{\kappa}(t-s),\end{equation} 
where $\vecc{\kappa}(t-s)$ is the memory kernel whose formula is given in \eqref{ham_memory}. Later, we will generalize the resulting covariance of the process $\vecc{\xi}(t)$ to an integral expression. 
We remark that the memory function $\vecc{\kappa}(t)$ and  the ``color" of the noise $\vecc{\xi}(t)$ are determined by the bath spectrum and the system-bath coupling. 

Now we pass to the continuum limit by replacing the sum over $k$ in $\vecc{\kappa}(t)$ by an integral $\int_{\RR^{+}} d\omega n(\omega)$, where $n(\omega)$ is a density of states.  Then, if the $\vecc{c}_k$ are replaced by $\vecc{c}(\omega) \in \RR^{r \times p}$, the memory function $\vecc{\kappa}(t)$ becomes the function:
\begin{equation}\vecc{\kappa}(t) = \int_{\RR^{+}} d\omega n(\omega) \frac{\vecc{c}(\omega) \vecc{c}(\omega)^{*}}{\omega^2} \cos(\omega t), \end{equation}
where $\hat{\vecc{\kappa}}_c(\omega) :=  n(\omega) \vecc{c}(\omega) \vecc{c}(\omega)^*/\omega^2 \in L^1(\RR^+)$.  Repeating the same procedure for the noise process and also replacing the $\vecc{x}'_k(0)$ and $\vecc{p}_k(0)$ by $\vecc{x}'(\omega)$ and $\vecc{p}(\omega)$ respectively, $\vecc{\xi}(t)$ becomes:
 \begin{equation} \label{limit_noise}
 \vecc{\xi}(t) = \int_{\RR^+} d\omega n(\omega) \vecc{c}(\omega) \left(  \vecc{x}'(\omega)  \cos(\omega t) + \frac{\vecc{p}(\omega)}{ \omega} \sin(\omega t) \right). \end{equation}

The choice of the $n(\omega)$ and $\vecc{c}(\omega)$ specifies the memory function and therefore (by the fluctuation-dissipation relation) the statistical properties of the noise process. We write $\vecc{\kappa}(t)$ as an inverse Fourier transform of a measure:
\begin{equation} \label{memory}
\vecc{\kappa}(t) = \frac{1}{2\pi} \int_{\RR} \vecc{S}(\omega) e^{i \omega t} d\omega,  
\end{equation}
where the measure is absolutely continuous with respect to the Lebesgue measure, with the density
$\vecc{S}(\omega) = \pi \hat{\vecc{\kappa}}_c(\omega) \geq 0$. The density $\vecc{S}(\omega)$ is known as the {\it spectral density} of the bath. 

In the following examples, we take $n(\omega) =  2\omega^2/\pi$ (Debye-type spectrum for phonon bath).

\begin{example}
If we choose $\vecc{c}(\omega) \in \RR^{d \times d}$ to be the identity matrix $\vecc{I}$ multiplied by a scalar constant that is independent of $\omega$, then $\vecc{\kappa}(t)$ is proportional to $\delta(t) \vecc{I}$. This leads to a Langevin equation driven by white noise, in which the damping term is instantaneous. In this case, we have the SDE system for $(\vecc{x}_t, \vecc{v}_t) \in \RR^{d \times d}$: 
\begin{align}
d \vecc{x}_t &= \vecc{v}_t dt, \label{side_derived1} \\ 
m d \vecc{v}_t &= -\vecc{\nabla}_{\vecc{x}} U (\vecc{x}_t) dt - \vecc{g}(\vecc{x}_t) \vecc{g}^*(\vecc{x}_t)\vecc{v}_t dt + \vecc{g}(\vecc{x}_t) \vecc{\xi}_t dt, \label{side_derived2}  
\end{align}
where $\vecc{\xi}_t$ is a white noise.      
\end{example}


\begin{example} If we  choose $\vecc{c}(\omega) \in \RR^{d \times d}$ to be the diagonal matrix with the $k$th entry 
\begin{equation}
\frac{\alpha_k}{\sqrt{\alpha^2_k+\omega^2}} ,\end{equation} 
where the $\alpha_{k} > 0$, then we have:
\begin{equation} \vecc{\kappa}(t) = \vecc{A} e^{-\vecc{A}|t|}, \end{equation}
where $\vecc{A}$ is the constant diagonal matrix with the $k$th entry equal $\alpha_{k}$. This gives SIDE \eqref{side2}.  On the other hand, choosing $\vecc{c}(\omega)$ to be the diagonal matrix with the $k$th entry
\begin{equation}
\left(\frac{\omega_{kk}}{\tau_{kk}}\right)^2 \frac{1}{\sqrt{\omega^2 (\omega_{kk}^2/\tau_{kk})^2+(\omega^2-(\omega_{kk}/\tau_{kk})^2)^2}}\end{equation} allows us to obtain the covariance function of a harmonic noise process, where the $\omega_{kk}$ and $\tau_{kk}$ are the diagonal entries of the matrices $\vecc{\Omega}$ and $\vecc{\tau}$ respectively. In the general case where $\vecc{\kappa}(t)$ is of the form \eqref{memory_realized}, one may take $\vecc{M}_1 = \vecc{I}$, $\vecc{\Gamma}_1$ to be positive definite (so that the Lyapunov equation gives $\vecc{\Gamma}_1 = \vecc{\Sigma}_1 \vecc{\Sigma}_1^*/2$) and choose 
\begin{equation}
\vecc{c}(\omega) = \frac{1}{\sqrt{2}} \vecc{C}_1 (\vecc{\Gamma}^2_1 +  \omega^2 \vecc{I})^{-1/2} \vecc{\Sigma}_1.\\
\end{equation}
\end{example}

\noindent {\bf Conclusion.} The goal of this paper is to study effective dynamics of systems which  can be modeled by equations of the form \eqref{genle_general}. As argued above, Hamiltonian systems describing particles interacting with heat baths can be modeled by equations of the form \eqref{genle_general}. 


\section{Outline of the Proof of Theorem \ref{skthm}} \label{proof_sketch}
We provide minimal outline of the proof of Theorem \ref{skthm} in the following.

\subsection{Derivation of Limiting SDE} 
We start by rewriting SDE \eqref{gsk2} as
\begin{equation}\vecc{v}^m_{t} dt = - m \boldsymbol{\gamma}^{-1}(\vecc{x}^m_{t})  d\vecc{v}^m_{t} + \boldsymbol{\gamma}^{-1}(\vecc{x}^m_{t}) \vecc{F}(\vecc{x}^m_{t}) dt + \boldsymbol{\gamma}^{-1}(\vecc{x}^m_{t}) \vecc{\sigma}(\vecc{x}^m_{t}) d\vecc{W}_{t}.\end{equation}

The integral form of the above is given by
\begin{align}
\vecc{x}^m_{t} &= \vecc{x} - m\int_{0}^{t} \boldsymbol{\gamma}^{-1}(\vecc{x}^m_{s})  d\vecc{v}^m_{s} + \int_{0}^{t} \boldsymbol{\gamma}^{-1}(\vecc{x}^m_{s}) \vecc{F}(\vecc{x}^m_{s}) ds + \int_{0}^{t} \boldsymbol{\gamma}^{-1}(\vecc{x}^m_{s}) \vecc{\sigma}(\vecc{x}^m_{s}) d\vecc{W}_{s}.
\end{align}

We are interested in the limit as $m \to 0$ of the process $\vecc{x}^m_{t}$. As $m \to 0$, we expect the sum of the second and third integral terms in the right hand side above to converge to $\int_0^t \vecc{\gamma}^{-1}(\vecc{X}_s) \vecc{F}(\vecc{X}_s) ds + \int_0^t \vecc{\gamma}^{-1}(\vecc{X}_s) \vecc{\sigma}(\vecc{X}_s) d\vecc{W}_s$. 

To examine the asymptotics of the first integral term when $m$ becomes small, we integrate by parts to write its $i$th component  as:
\begin{align} 
\int_{0}^{t} (\gamma^{-1})_{ij}(\vecc{x}^m_{s}) d(m v^m_{s})_{j} &= (\gamma^{-1})_{ij}(\vecc{x}^m_{t})  m (v^m_{t})_{j} - (\gamma^{-1})_{ij}(\vecc{x})  m (v^m)_{j} \nonumber
\\ 
&\ \ \ \ \ \ - \int_{0}^{t} \frac{\partial}{\partial x^m_{l}}[(\gamma^{-1})_{ij}(\vecc{x}^m_{s})] m (v^m_{s})_{j} (v^m_{s})_{l} ds.
\end{align}
Note that the product $ m (v^m_{s})_{j} (v^m_{s})_{l}$ is the $(j,l)$-entry of the matrix $m \vecc{v}^m_{s} (\vecc{v}^m_{s})^{*}$. 


We now examine the asymptotic behavior of the above expression in the limit as $m \to 0$.
Following \cite{hottovy2015smoluchowski}, we express the matrix $m \vecc{v}^m_{s} (\vecc{v}^m_{s})^{*}$, $s \in [0,t]$, as a solution to an equation by applying It\^o's formula to the matrix $(m \vecc{v}^m_{s})(m (\vecc{v}^m_{s})^{*})$. This leads to:
\begin{align}
d[(m \vecc{v}^m_{s})(m (\vecc{v}^m_{s})^{*})] &= -[\vecc{\gamma}(\vecc{x}^m_{s}) (m \vecc{v}^m_{s} (\vecc{v}^m_{s})^{*} ds) + (m \vecc{v}^m_{s} (\vecc{v}^m_{s})^{*} ds) \vecc{\gamma}^{*}(\vecc{x}^m_{s})] + \vecc{\sigma}(\vecc{x}^m_{s}) \vecc{\sigma}^{*}(\vecc{x}^m_{s}) ds + d\vecc{U}^m_{s} + d(\vecc{U}^m_s)^{*},
\end{align}
where 
\begin{align}
d\vecc{U}^m_{s} &= (\vecc{F}(\vecc{x}^m_{s}) ds + \vecc{\sigma}(\vecc{x}^m_{s}) d\vecc{W}_{s}) m (\vecc{v}^m_{s})^{*}, \\ 
d(\vecc{U}^m_s)^{*} &= m \vecc{v}^m_{s}(\vecc{F}^{*}(\vecc{x}^m_{s}) ds + d\vecc{W}_{s}^{*} \vecc{\sigma}^{*}(\vecc{x}^m_{s})).
\end{align}

Denoting $m \vecc{v}^m_{s} (\vecc{v}^m_{s})^{*} ds$ by $\vecc{V}$, $-\vecc{\gamma}(\vecc{x}^m_{s})$ by $\vecc{Q}$ and letting \begin{equation}\vecc{C} := d[(m \vecc{v}^m_{s})(m (\vecc{v}^m_{s})^{*})] - \vecc{\sigma}(\vecc{x}^m_{s}) \vecc{\sigma}^{*}(\vecc{x}^m_{s}) ds - d\vecc{U}^m_{s} - d(\vecc{U}^m_{s})^{*}, \end{equation} we can write the above equation as the following Lyapunov equation
\begin{equation}\vecc{Q} \vecc{V} + \vecc{V} \vecc{Q}^{*} = \vecc{C}. \end{equation}
By Assumption \ref{a2}, all  real parts of the eigenvalues of $\vecc{Q}$ are negative, thus the Lyapunov equation has a unique solution given by: 
\begin{equation}\vecc{V} = - \int_{0}^{\infty} e^{\vecc{Q} y} \vecc{C} e^{\vecc{Q}^{*}y} dy. \end{equation}

Writing this out explicitly, we obtain:
\begin{align}
m \vecc{v}^m_{s} (\vecc{v}^m_{s})^{*} ds &= - \int_{0}^{\infty} e^{-\vecc{\gamma}(\vecc{x}^m_{s}) y} d[(m \vecc{v}^m_{s})(m (\vecc{v}^m_{s})
^{*})]  e^{-\vecc{\gamma}^{*}(\vecc{x}^m_{s}) y} dy  \nonumber \\ 
&\ \ \ \ + \int_{0}^{\infty} e^{-\vecc{\gamma}(\vecc{x}^m_{s}) y} (\vecc{\sigma}(\vecc{x}^m_{s}) \vecc{\sigma}^{*}(\vecc{x}^m_{s}) ds) e^{-\vecc{\gamma}^{*}(\vecc{x}^m_{s}) y} dy  \nonumber \\
&\ \ \ \ + \int_{0}^{\infty} e^{-\vecc{\gamma}(\vecc{x}^m_{s}) y} (d\vecc{U}^m_{s}+ d(\vecc{U}^m_{s})^{*})   e^{-\vecc{\gamma}^{*}(\vecc{x}^m_{s}) y} dy.\label{estb}
\end{align}

Based on the prior result in \cite{hottovy2015smoluchowski}, we expect that only the second term on the right hand side has a nonzero limit as $m \to 0$. The other terms are expected to vanish as $m \to 0$. Thus, in the limit $m \to 0$, we expect that $m \vecc{v}^m_{t} (\vecc{v}^m_{t})^{*}$ converges to the solution, $\vecc{J}$, of the Lyapunov equation:
\begin{equation}\vecc{J} \vecc{\gamma}^{*} + \vecc{\gamma} \vecc{J} = \vecc{\sigma} \vecc{\sigma}^{*}, \end{equation}
 given in the statement of Theorem \ref{skthm} (see eqn. \eqref{lyp}). 

\subsection{Moment Estimates}
To justify the above convergence arguments, we provide  estimate on the $p$th moment of the momentum process $\vecc{p}^m_{t} := m\vecc{v}^m_{t}$, in the limit $m \to 0$. 

\begin{proposition} Suppose that Assumption \ref{a1}-\ref{a4} hold. For all $p\geq 1$, $T > 0$, there exists a positive random variable $m_1$ such that: 
\begin{equation} \mathbb{E}\left[ \sup_{t \in [0,T]} |\vecc{p}^m_{t}|^{p}; m \leq m_1 \right] \to 0 \end{equation}
as $m \to 0$.  
\end{proposition}
\begin{proof}
For $t \in [0,T]$, $m >0$,  the process $\vecc{p}_t^m$ satisfies the SDE:
\begin{equation}
d\vecc{p}_t^m = -\frac{\vecc{\gamma}(\vecc{x}_t^m)}{m} \vecc{p}_t^m dt + \vecc{F}(\vecc{x}_t^m) dt + \vecc{\sigma}(\vecc{x}_t^m) d\vecc{W}_t.
\end{equation} 
Let $\tau \in [0,t]$ and rewrite the above equation as:
\begin{align}
d\vecc{p}_t^m &= -\frac{\vecc{\gamma}(\vecc{x}_\tau^m)}{m} \vecc{p}_t^m dt + \frac{1}{m} \left( \vecc{\gamma}(\vecc{x}_{\tau}^m)  - \vecc{\gamma}(\vecc{x}_t^m) \right) \vecc{p}_t^m  dt +  \vecc{F}(\vecc{x}_t^m) dt + \vecc{\sigma}(\vecc{x}_t^m) d\vecc{W}_t,
\end{align} 
which admits the following solution representation:
\begin{align}
\vecc{p}_t^m &= e^{-\frac{\vecc{\gamma}(\vecc{x}_\tau^m)}{m}t} \vecc{p}_0^m + \int_0^t e^{-\frac{\vecc{\gamma}(\vecc{x}_\tau^m)}{m}(t-s)} \vecc{F}(\vecc{x}_s^m) ds + \int_0^t e^{-\frac{\vecc{\gamma}(\vecc{x}_\tau^m)}{m}(t-s)} \vecc{\sigma}(\vecc{x}_s^m) d\vecc{W}_s \nonumber \\
&\ \ \ \  + \frac{1}{m} \int_0^t e^{-\frac{\vecc{\gamma}(\vecc{x}_\tau^m)}{m}(t-s)} (\vecc{\gamma}(\vecc{x}_\tau^m)-\vecc{\gamma}(\vecc{x}_s^m)) \vecc{p}_s^m ds.    
\end{align}

We set $\tau = t$ in the above representation to obtain:
\begin{align}
\vecc{p}_t^m &= e^{-\frac{\vecc{\gamma}(\vecc{x}_t^m)}{m}t} \vecc{p}_0^m + \int_0^t e^{-\frac{\vecc{\gamma}(\vecc{x}_t^m)}{m}(t-s)} \vecc{F}(\vecc{x}_s^m) ds + \int_0^t e^{-\frac{\vecc{\gamma}(\vecc{x}_t^m)}{m}(t-s)} \vecc{\sigma}(\vecc{x}_s^m) d\vecc{W}_s \nonumber \\
&\ \ \ \  + \frac{1}{m} \int_0^t e^{-\frac{\vecc{\gamma}(\vecc{x}_t^m)}{m}(t-s)} (\vecc{\gamma}(\vecc{x}_t^m)-\vecc{\gamma}(\vecc{x}_s^m)) \vecc{p}_s^m ds.    
\end{align}
Therefore, 
\begin{align}
|\vecc{p}_t^m| &\leq \left|e^{-\frac{\vecc{\gamma}(\vecc{x}_t^m)}{m}t} \vecc{p}_0^m\right| + \int_0^t \left\|e^{-\frac{\vecc{\gamma}(\vecc{x}_t^m)}{m}(t-s)}\right\| \cdot |\vecc{F}(\vecc{x}_s^m)| ds + \left| \int_0^t e^{-\frac{\vecc{\gamma}(\vecc{x}_t^m)}{m}(t-s)} \vecc{\sigma}(\vecc{x}_s^m) d\vecc{W}_s \right| \nonumber \\
&\ \ \ \  + \sup_{s \in [0,T]} |\vecc{p}_s^m| \cdot \frac{1}{m} \int_0^t \left\| e^{-\frac{\vecc{\gamma}(\vecc{x}_t^m)}{m}(t-s)} \right\| \cdot \|(\vecc{\gamma}(\vecc{x}_t^m)-\vecc{\gamma}(\vecc{x}_s^m)) \|  ds.    
\end{align}
By assumption on the boundedness and spectrum of $\vecc{\gamma}(\vecc{x}_\tau) \in \RR^{n \times n}$ (see Assumption \ref{a1}-\ref{a2}), there exists positive constants $\kappa>0$ and $C>0$ such that
\begin{equation}
\left\|e^{-\frac{\vecc{\gamma}(\vecc{x}^m_\tau)}{m}s}\right\| \leq C e^{-\frac{\kappa}{m}s},
\end{equation}
for all $s, \tau \in [0,T]$, $m>0$. Indeed, applying the formula (A.2.4) in \cite{kabanov2013two}, one has:
\begin{equation}
\left\|e^{-\frac{\vecc{\gamma}(\vecc{x}^m_\tau)}{m}s}\right\| \leq e^{-\frac{\Lambda}{m}s} \left(1 + 2 \|\vecc{\gamma}\| \sum_{k=1}^{n-1} \frac{1}{k!} \left(\frac{2 s \|\vecc{\gamma}\|}{m}\right)^k \right),
\end{equation}
where $\Lambda := \min_k Re(\lambda_k) > 0$ and the $\lambda_k$ are the eigenvalues of $\vecc{\gamma}$. 
Therefore, there exists a constant $C>0$ such that $\left\|e^{-\frac{\vecc{\gamma}(\vecc{x}^m_\tau)}{m}s}\right\| \leq C e^{-\frac{\Lambda}{2m}s}$. 

Using this, we obtain the following $\mathbb{P}$-a.s. estimate:
\begin{align}
\sup_{t \in [0,T]} |\vecc{p}_t^m| &\leq C  \left| \vecc{p}_0^m\right| + C  \sup_{t \in [0,T]} \int_0^t  e^{-\frac{\kappa}{m}(t-s)} |\vecc{F}(\vecc{x}_s^m)| ds + \sup_{t \in [0,T]} \left| \int_0^t e^{-\frac{\vecc{\gamma}(\vecc{x}_t^m)}{m}(t-s)} \vecc{\sigma}(\vecc{x}_s^m) d\vecc{W}_s \right| \nonumber \\
&\ \ \ \  + \sup_{t \in [0,T]} |\vecc{p}_t^m| \left( \sup_{t \in [0,T]} \frac{\tilde{C}}{m} \int_0^t  e^{-\frac{\kappa}{m}(t-s)} \|\vecc{\gamma}(\vecc{x}_t^m)-\vecc{\gamma}(\vecc{x}_s^m) \|  ds \right),    \label{as_bound}
\end{align}
where $\tilde{C}>0$ is a constant. 

Next, the key observation on \eqref{as_bound} is that the term  in parenthesis in the second line above can be made  small by choosing a sufficiently small $m$. More precisely, using an adapted version of Lemma A.2.4 in \cite{kabanov2013two}, there exists a (generally random) $m_1 > 0$ such that for $m \leq m_1$, 
\begin{align}
\sup_{t \in [0,T]} |\vecc{p}_t^m| &\leq C  \left| \vecc{p}_0^m\right| + C  \sup_{t \in [0,T]} \int_0^t  e^{-\frac{\kappa}{m}(t-s)} |\vecc{F}(\vecc{x}_s^m)| ds + \sup_{t \in [0,T]} \left| \int_0^t e^{-\frac{\vecc{\gamma}(\vecc{x}_t^m)}{m}(t-s)} \vecc{\sigma}(\vecc{x}_s^m) d\vecc{W}_s \right| \nonumber \\
&\ \ \ \  + \frac{1}{2} \sup_{t \in [0,T]} |\vecc{p}_t^m|, 
\end{align}
and so:
\begin{align}
\sup_{t \in [0,T]} |\vecc{p}_t^m| &\leq 2C  \left| \vecc{p}_0^m\right| + 2C  \sup_{t \in [0,T]} \int_0^t  e^{-\frac{\kappa}{m}(t-s)} |\vecc{F}(\vecc{x}_s^m)| ds + 2\sup_{t \in [0,T]} \left| \int_0^t e^{-\frac{\vecc{\gamma}(\vecc{x}_t^m)}{m}(t-s)} \vecc{\sigma}(\vecc{x}_s^m) d\vecc{W}_s \right| \\ 
&\leq 2C m \left| \vecc{v}_0^m\right| + 2C \frac{m}{\kappa}   \sup_{\vecc{u} \in \RR^n} |\vecc{F}(\vecc{u})| + 2\sup_{t \in [0,T]} \left| \int_0^t e^{-\frac{\vecc{\gamma}(\vecc{x}_t^m)}{m}(t-s)} \vecc{\sigma}(\vecc{x}_s^m) d\vecc{W}_s \right|.
\end{align}
Therefore, for $p \geq 1$,
\begin{align} 
\mathbb{E}\left[\sup_{t \in [0,T]} |\vecc{p}_t^m|^p; m \leq m_1 \right] &\leq C_1(p)  \mathbb{E}\left[ \left|m \vecc{v}_0^m\right|^p; m \leq m_1 \right]  + C_2(p) \mathbb{E}\left[ m^p; m \leq m_1 \right] \nonumber \\ 
&\ \ \ \  + C_3(p) \mathbb{E} \left[  \sup_{t \in [0,T]} \left| \int_0^t e^{-\frac{\vecc{\gamma}(\vecc{x}_t^m)}{m}(t-s)} \vecc{\sigma}(\vecc{x}_s^m) d\vecc{W}_s \right|^p; m \leq m_1 \right], \label{estt} 
\end{align}
where $C_1(p)$, $C_2(p)$ and $C_3(p)$ are some positive constants.

We estimate the last term in the above, using in particular the Burkholder-Davis-Gundy inequality \cite{karatzas2012Brownian}:
\begin{align}
\mathbb{E} \left[  \sup_{t \in [0,T]} \left| \int_0^t e^{-\frac{\vecc{\gamma}(\vecc{x}_t^m)}{m}(t-s)} \vecc{\sigma}(\vecc{x}_s^m) d\vecc{W}_s \right|^p; m \leq m_1 \right] &\leq \mathbb{E} \left[  \sup_{t \in [0,T]} \left| \int_0^t e^{-\frac{\vecc{\gamma}(\vecc{x}_t^m)}{m}(t-s)} \vecc{\sigma}(\vecc{x}_s^m) d\vecc{W}_s \right|^p\right]  \\ 
&\leq C(p,n) \mathbb{E}  \left(\int_0^T \|  e^{-\frac{\vecc{\gamma}(\vecc{x}_t^m)}{m}(t-s)} \vecc{\sigma}(\vecc{x}_s^m) \|_{F}^2 ds \right)^{p/2} \\ 
&\leq \tilde{C}(p,n) \mathbb{E} \left(\int_0^T  e^{-2\kappa(t-s)/m} ds \right)^{p/2} \\
&\leq \frac{\tilde{C}(p,n)}{\kappa} m^{p/2}, 
\end{align}
where $C(p,n)$, $\tilde{C}(p,n)$ are positive constants dependent on $p$ and $n$, and $\|\cdot \|_F$ denotes Frobenius norm. 
Using this estimate, \eqref{estt} and Assumption \ref{a3}, we see that  $\mathbb{E}\left[\sup_{t \in [0,T]} |\vecc{p}_t^m|^p; m \leq m_1 \right] \to 0$ as $m \to 0$. 
\end{proof}



We also need the following estimate on a class of integrals with respect to products of the components of the momentum process $\vecc{p}^m_{t} = m\vecc{v}^m_{t}$. The estimate is a straightforward modification  of the one given in Proposition 2.3 in \cite{birrell2017homogenization}.

\begin{proposition} 
Suppose that Assumption \ref{a1}-\ref{a4} hold.
Let $h :\RR^{n} \to \RR$ be a $C_{b}^{1}$ function (i.e. continuously differentiable and bounded function) on $[0,T]$, with bounded first derivative $\vecc{\nabla}_{\vecc{x}} h(\vecc{x})$ for every $\vecc{x} \in \RR^n$. Then for any $p \geq 1$, $T>0$, $i,j = 1, \dots, n$,
\begin{equation} \mathbb{E} \left[ \sup_{t \in [0,T]} \left| \int_{0}^{t} h(\vecc{x}^m_{s}) d((p^m_{s})_{i} (p^m_{s})_{j}) \right|^p; m \leq m_1\right]  \to 0 \end{equation}  as $m \to 0$, where the $m_1$ is from Proposition 1. Here $i,j$ denote the components of the momentum process $\vecc{p}^m_{t}$ in the standard basis for $\RR^{n}$. 
\end{proposition}



Using the above moment estimates and the proof techniques (the main tools are well known ordinary and stochastic integral inequalities as well as a Gronwall type argument) in \cite{birrell2017small, birrell2017homogenization}, one  obtain the convergence of $\vecc{x}^m_{t}$ to $\vecc{X}_{t}$ in the limit as $m \to 0$ in the following sense:   for all finite $T>0$, $p \geq 1$, 
\begin{equation} \label{nonst}
\mathbb{E}\left[ \sup_{t \in [0,T]} |\vecc{x}_t^m - \vecc{X}_t|^p; m \leq m_1\right] \to 0,
\end{equation}
as $m \to 0$, where the $m_1$ is from Proposition 1. This implies that for all finite $T>0$,  $\sup_{t \in [0,T]} |\vecc{x}_t^m - \vecc{X}_t| \to 0$ in probability, in the limit as $m \to 0$ (see Lemma \ref{lemm} below).


\begin{lemma}  \label{lemm}
Let $Y_m \ge 0$ ($m \ge 0$) be a family of random variables.  Suppose that there exists a strictly positive random variable $m^*(\omega)$ such that 
$$
\lim_{m \to 0}\mathbb{E} \left[Y_m^p; m^* \ge m\right] \to 0
$$
for some $p \ge 1$.  Then $Y_m \to 0$ in probability as $m \to 0$.  
\end{lemma}
\begin{proof}
Let $g$ be a bounded continuous function on $\RR_+$.  We have
\begin{equation}
\mathbb{E} \left[g(Y_m)\right] = \mathbb{E} \left[g(Y_m); m^* \le m\right] + \mathbb{E} \left[g(Y_m); m^* \ge m\right].
\end{equation}
The first term is bounded from above by $(\sup |g|) \mathbb{P}\left[m^* \le m\right]$ which goes to zero as $m \to 0$, since $\cap_{m > 0}\{m^* \le m\} = \emptyset$.  To estimate the second term, let $\epsilon > 0$.  Choose $\delta > 0$ such that $|g(y)| < \epsilon$ whenever $y < \delta$.  We have
\begin{equation}
\mathbb{E} \left[g(Y_m); m^* \ge m\right] \le \mathbb{E} \left[g(Y_m); m^* \ge m, Y_m < \delta\right] +  \mathbb{E} \left[g(Y_m); m^* \ge m, Y_m \ge \delta\right].
\end{equation}
The first term is bounded from above by $\epsilon$.  The function ${g(y) \over y}$ on $\{y: y\ge \delta\}$ is bounded by some constant $M$, so the second term on the right-hand side of the above equation is bounded from above by $M \mathbb{E} \left[Y_m; m^* \ge m\right]$ which goes to zero as $m \to 0$ by assumption.  The lemma is proven.
\end{proof}

We end this appendix with a remark on the mode of convergence stated in Theorem 1.

\begin{remark}
Provided that  for $p \geq 1$,  there exists a constant $C>0$ such that $\mathbb{E} \sup_{t \in [0,T]} |\vecc{p}_t^m|^p < C$ for all $m>0$ (so that the family of random variables $(\sup_{t \in [0,T]} |\vecc{p}_t^m|^p)_{m > 0}$ is uniformly integrable), one could, using Proposition 1, Lemma 1 and Theorem 13.7 in \cite{williams1991probability}, obtain $L^p$-convergence of $\sup_{t \in [0,T]} |\vecc{p}_t^m|$ to zero and hence strengthen the convergence result stated in Theorem 1 to $L^p$-convergence. The uniform integrability condition is satisfied if $\mathbb{E} \sup_{t \in [0,T]} \|\vecc{\Phi}^m(t)\|^p$ is bounded uniformly in $m$, where $\vecc{\Phi}^m(t)$ is the fundamental matrix that solves the random initial value problem: 
\begin{equation}
\frac{\partial}{\partial t} \vecc{\Phi}^m(t) = -\frac{\vecc{\gamma}(\vecc{x}_t^m)}{m} \vecc{\Phi}^m(t),  \ \ \vecc{\Phi}^m(0) = \vecc{I}, \ \  t \in [0,T].
\end{equation}
However, it is not obvious how one could verify the latter condition from our assumptions on $\vecc{\gamma}$. Roughly speaking, one does not have a good control of $\sup_{t \in [0,T]} |\vecc{p}_t^m|^p$ outside of the set $\{m \leq m_1\}$. If $\vecc{\gamma}$ was, in addition, symmetric (and so all the (real) eigenvalues of $\vecc{\gamma}$ are bounded from below by a positive constant --  c.f. \cite{ birrell2017homogenization}),  the condition can be easily verified. 
\end{remark}

\end{document}